\documentclass[11pt,times]{article}

\usepackage[utf8]{inputenc}
\usepackage{amssymb,amsmath}
\usepackage{algorithm}
\usepackage{algpseudocode}
\usepackage{amsthm}
\usepackage{enumerate}
\usepackage[font={scriptsize,it}]{caption}
\usepackage{subcaption}
\allowdisplaybreaks
\usepackage[T1]{fontenc}

\usepackage{multirow}
\usepackage{bm}

\usepackage{float}
\newfloat{algorithm}{t}{lop}

\usepackage{tikz}
\usetikzlibrary{decorations.pathmorphing}
\usetikzlibrary{shapes,fit}
\usetikzlibrary{backgrounds}

\algtext*{EndWhile}
\algtext*{EndFor}
\algtext*{EndIf}

\newtheorem{theorem}{Theorem}[section]
\newtheorem{corollary}[theorem]{Corollary}
\newtheorem{definition}[theorem]{Definition}
\newtheorem{lemma}[theorem]{Lemma}

\newtheorem*{theorem*}{Theorem}

\newtheorem*{corollary*}{Corollary}
\newtheorem*{lemma*}{Lemma}
\newtheorem*{observation*}{Observation}

    \newcommand{\Xomit}[1]{}

\newcommand{\R}{Randomized }
\newcommand{\A}{Arbitrary }
\newcommand{\I}{Integer }
\newcommand{\D}{Deterministic }
\newcommand{\DI}{Directed \& Undirected }

    \newcommand{\congest}{{\sc Congest}}
\newcommand{\Tau}{H}
\newcommand{\hide}[1]{}
\newcommand{\centers}{blocker set }

\newcommand{\boi}{\begin{enumerate}}
\newcommand{\eoi}{\end{enumerate}}
\newcommand{\bii}{\begin{itemize}}
\newcommand{\eii}{\end{itemize}}

\advance \topmargin by -\headheight
\advance \topmargin by -\headsep
\evensidemargin \oddsidemargin
\begin{document}

\title{Faster Deterministic All Pairs Shortest Paths in {\sc Congest} Model}

\author{Udit Agarwal $^{\star}$ and Vijaya Ramachandran\thanks{University of Texas, Austin TX 78712. Email: {\tt udit@utexas.edu, vlr@cs.utexas.edu}. }}

\maketitle

\begin{abstract}

 We present a new deterministic algorithm for distributed weighted 
 all pairs shortest paths (APSP)
 in both undirected and directed graphs. Our algorithm runs in $\tilde{O}(n^{4/3})$ rounds in the {\sc Congest} models on graphs with arbitrary edge weights, and it improves on the previous $\tilde{O}(n^{3/2})$ bound of Agarwal et al.~\cite{ARKP18}. The main components of our new algorithm are a new faster technique for constructing blocker set deterministically and a new pipelined method for deterministically propagating distance values from source nodes 
 to the blocker set nodes in the network. Both of these techniques have potential applications to other distributed algorithms.

\vspace{.05in}

Our new deterministic algorithm for computing blocker set adapts the NC approximate 
hypergraph set cover algorithm 
in~\cite{BRS94} to the distributed construction of a blocker set. It follows the two-step process of first designing a randomized algorithm that uses only pairwise independence, and then derandomizes this algorithm
using a sample space of linear size.
This algorithm runs in almost the same number of rounds as the initial step in our APSP algorithm that computes $h$-hops shortest paths.  This result significantly improves on the deterministic blocker set algorithms 
in~\cite{ARKP18, AR19} by removing an additional $n\cdot |Q|$ term in the round bound, where $Q$ is the blocker set. 

\vspace{.05in}

The other new component in our APSP algorithm is a deterministic pipelined approach to propagate distance values from source nodes 
to blocker nodes. We use a simple natural round-robin method for this step, and we show using a suitable progress measure that it achieve the $\tilde{O}(n^{4/3})$ bound on the number of rounds. It appears that the standard deterministic methods for efficiently broadcasting multiple values, and for sending  or receiving messages using the routing schedule in an undirected APSP 
algorithm~\cite{HPDG+19,LSP19} do not apply to this setting.
\end{abstract}

~
~

\section{Introduction}	\label{sec:intro}

We study the computation of all pairs shortest path (APSP) in the widely-used \textsc{Congest} model of distributed computing
(see, e.g., \cite{ARKP18,Elkin17,HNS17,LP13}).
In the {\sc Congest}  model (described in Section~\ref{sec:congest}), the 
input is a directed (or undirected) graph $G=(V,E)$ and the distributed computation occurs at the nodes in this graph. The
output of the APSP problem is to compute at each node $v$, the shortest path distances  from every 
source 
\footnote{We will refer to a vertex as a source if we compute shortest paths from that vertex.}
node to $v$
in the network.
We assume an  arbitrary non-negative weight on each edge,
and we let $|V| =n$. In this paper we consider the computation of exact (and not approximate) shortest paths.

\noindent
{\bf Overview of our contributions.}
In this paper we present a
$\tilde{O}(n^{4/3})$ round deterministic algorithm for the weighted APSP problem.
Table~\ref{table} lists earlier results for this problem~\cite{Elkin17,HNS17,ARKP18,AR19,BN19}.
All of these results as well as our new result can handle zero weight edges, and these algorithms are qualitatively 
different from algorithms for unweighted APSP.

Our  algorithm follows the general 3-phase strategy initiated by Ullman and Yannakakis~\cite{UY91} for parallel computation of path problems in directed graphs:

\begin{enumerate}
\item  Compute $h$-hop shortest paths for each source for a suitable value of $h$.
(An $h$-hop path is a path that contains at most $h$ edges.)

\item  Find a small {\it blocker set} $Q$ that intersects all $h$-hop
paths computed in Step 1. (With randomization, this step is very simple: 
a random sample of the vertices of size $O ((n/h) \cdot  \log n)$ satisfies this property w.h.p. in $n$.)

\item Compute shortest paths between all pairs of vertices within $Q$. Then, use this information and the $h$-hop trees from Step 1 in a suitable algorithm to compute the APSP output at each node in $V$.
\end{enumerate}

{\sc Congest} directed APSP algorithms that fall in this framework include the randomized algorithm in Huang et al.~\cite{HNS17}  
that runs in $\tilde{O}(n^{5/4})$
rounds
 for polynomial integer edge-weights, the deterministic algorithm in Agarwal et al.~\cite{ARKP18} that runs in $\tilde{O}(n^{3/2})$
rounds
  for arbitrary edge-weights, and the deterministic algorithm in Agarwal and Ramachandran~\cite{AR19} that improves on~\cite{ARKP18}
for moderate integer edge-weights.

\begin{table}
\scriptsize
\centering
\caption{Table comparing our new results with previous known results for exact weighted APSP problem.} \label{table}
\def\arraystretch{1.5}
{\setlength{\tabcolsep}{.3em}
\begin{tabular}{| c | c | c | c | c |}
\hline
Author & Arbitrary/ Integer & Randomized/ & Undirected/ & Round  \vspace{-.05in}\\
 & weights  & Deterministic & (Directed \& Undirected) & Complexity \\
 \hline
Huang et al.~\cite{HNS17} & Integer  & \R & Directed \& Undirected & $\tilde{O}(n^{5/4})$ \\
\hline
Elkin~\cite{Elkin17} & Arbitrary  & \R & Undirected & $\tilde{O}(n^{5/3})$ \\
\hline
Agarwal et al.~\cite{ARKP18} & \A  & \D & Directed \& Undirected & $\tilde{O}(n^{3/2})$ \\
\hline
 & \multirow{2}{*}{\I}  & \multirow{2}{*}{\D} & \multirow{2}{*}{Directed \& Undirected} & $\tilde{O}(n^{5/4} \cdot W^{1/4})$ ($W=$ max edge wt) \\
Agarwal \& &  & &  & $\tilde{O}(n \cdot \Delta^{1/3})$ ($\Delta$= max SP distance) \\\cline{2-5}
 Ramachandran~\cite{AR19} & \A  & \R & \DI & $\tilde{O}(n^{4/3})$ \\
\hline
Bernstein \& Nanongkai~\cite{BN19} & Arbitrary & Randomized & Directed \& Undirected & $\tilde{O}(n)$ \\
\hline
 \textbf{This Paper} &  \textbf{Arbitrary} & \textbf{\D} & \textbf{Directed \& Undirected} & $\bm{\tilde{O}(n^{4/3})}$ \\
\hline  
\end{tabular}}
\end{table}

Our new deterministic algorithm directly improves on~\cite{ARKP18}. The algorithm in~\cite{ARKP18} computes Step 1 in $O(n \cdot h)$ rounds by running the distributed Bellman-Ford
algorithm for $h$ hops from each source.
 Our algorithm leaves Step 1 unchanged from~\cite{ARKP18}, but it improves on both Step 2 and Step 3.
Our improved methods for implementing Steps 2 and 3 are the main technical contributions of this paper.
We list them below, followed by an  informal description of each of them.

\begin{enumerate}
\item A new deterministic algorithm for computing backer set (Algorithms~\ref{algRandBlocker}, \ref{algDet}). The ideas are derived from the Berger et al.'s NC algorithm for
 finding a small set cover in a hypergraph~\cite{BRS94},
and the result is an algorithm that significantly improves on the blocker set algorithm in Agarwal et al.~\cite{ARKP18} by removing an additional $n\cdot |Q|$ term in the round bound, where $Q$ is the blocker set.

\item A deterministic pipelined algorithm for propagating distance values from source nodes to blocker nodes (Algorithm~\ref{algCase2}). This algorithm deterministically propagates $\tilde{O}(n^{5/3})$ distance values from 
$n$ sources to $\tilde{O}(n^{2/3})$ blocker nodes in $\tilde{O}(n^{4/3})$ rounds, when the congestion at any node is at most $\tilde{O}(n^{4/3})$. Prior to this work, no deterministic algorithm was known that can implement this step in less than $n^{5/3}$ rounds.
\end{enumerate}

\subparagraph*{\bf (1) Deterministic Algorithm for Computing a Blocker Set (Step 2).}
For Step 2, \cite{ARKP18} gives a deterministic  algorithm that greedily chooses vertices to add to $Q$ at the cost of $O(n)$ rounds per vertex added,
for the cleanup cost for removing paths that are covered by this newly chosen vertex; this is 
 after an initial start-up cost of $O(n \cdot h)$. This gives
an overall cost of $O(nh + nq)$ for Step 2, where $q = |Q|= O((n/h) \cdot  \log n)$. 
Our new contribution is to construct $Q$ in a sequence of $polylog (n)$ steps, where each step adds several vertices to $Q$.
Our method incurs a cleanup 
cost of $O(|S|\cdot h)$ rounds per step after an initial start-up cost of $O(|S|\cdot h)$ rounds for an arbitrary source set $S$,
thereby removing the dependence on $q$ from this bound. 
($S=V$ gives the standard setting used in previous APSP algorithms.)
We achieve this by framing the computation of a small blocker set as an approximate set cover problem on a related hypergraph.
We then adapt the efficient NC algorithm in Berger et al.~\cite{BRS94} for computing an approximate minimum  set cover in a hypergraph to an $\tilde{O}(|S| \cdot h)$-round {\sc Congest} algorithm.  As in~\cite{BRS94}
this involves two main parts.
We first give  a randomized $\tilde{O} (|S|\cdot h)$-round algorithm that computes a blocker set of expected size $\tilde{O}(n/h)$ using only pairwise independent random variables. We then
derandomize this algorithm, again with an $\tilde{O} (|S|\cdot h)$-round algorithm.

\subparagraph*{\bf (2) Deterministic Pipelined Algorithm for Propagating Distance Values (Step 3).}
For Step 3, \cite{ARKP18} gave a deterministic  $O(n \cdot q)$- round algorithm, and  \cite{HNS17} gave a randomized
$\tilde{O} (n \cdot \sqrt q + n\cdot \sqrt{h})$-round algorithm.
 We replace the $n\cdot \sqrt{h}$ randomized algorithm used in~\cite{HNS17}  with a simple $n\cdot h$ round algorithm (similar to Step 1).
The randomized $O(n \cdot \sqrt q)$  method in~\cite{HNS17} computes
the {\it reversed $q$-sink shortest paths} problem that appears to use randomization in a crucial manner, 
by invoking
the randomized scheduling result of Ghaffari~\cite{Ghaffari15}, which allows
multiple algorithms to run concurrently
 in $O(d + c \cdot \log n)$
rounds, where $d$ bounds the {\it dilation} of any of the concurrent algorithms and $c$ bounds the 
{\it congestion} on any edge when considering all algorithms.
It is known that this result in~\cite{Ghaffari15} cannot be derandomized in a completely general setting. For Step 3, our contribution is to give a deterministic $\tilde{O} (n \cdot \sqrt q)$-round algorithm for the 
reversed $q$-sink shortest paths problem.
Our algorithm uses a simple round-robin pipelined approach.
To obtain the desired round bound we rephrase 
 the algorithm to work
 in {\it frames} which allows us to establish suitable progress in the pipelining to show that it terminates in
$\tilde{O}(n \cdot \sqrt q)$ rounds. 
We note that the standard known results on efficiently broadcasting multiple values, and on sending  or receiving messages using the routing schedule in an undirected APSP
 algorithm~\cite{HPDG+19,PR18,LSP19}
 do not apply to this setting.
 
Finally, we obtain the $\tilde{O}(n^{4/3})$ bound on the number of rounds by balancing the $\tilde{O}(nh)$ bound for Steps 1 and 2 with the $\tilde{O}(n \cdot \sqrt q)$ bound for 
the reversed $q$-sink shortest path problem,
as stated in the following theorem.

\begin{theorem}\label{thm:main}
There is a deterministic distributed algorithm that computes APSP on an $n$-node graph with arbitrary nonnegative 
edge-weights, directed or undirected, in $\tilde{O}(n^{4/3})$ rounds.
\end{theorem}

Theorem~\ref{thm:main} improves on prior results for deterministic APSP on weighted graphs in the {\sc Congest} model. If 
randomization is allowed, the very recent result in~\cite{BN19} gives an~$\tilde{O}(n)$-round randomized algorithm, which is
close to the known lower bound of $\Omega(n)$ rounds~\cite{CKP17},
that holds
 even for unweighted APSP.

 \begin{algorithm*}
\caption{Overall APSP Algorithm}
\begin{algorithmic}[1]
\Statex Input: number of hops $h = n^{1/3}$
\State Compute $h$-CSSSP for set $V$ using the algorithm in~\cite{AR19}. \label{computeCSSSP}
\State Compute a blocker set $Q$ of size $\tilde{O}(n/h)$ for the $h$-CSSSP computed in Step~\ref{computeCSSSP} (described in Section~\ref{sec:blocker}).	\label{constructQ}
\State \textbf{For each $\bm{c\in Q}$ in sequence:} Compute $h$-in-SSSP rooted at $c$. \label{computeInCSSSP}
\State \textbf{For each $\bm{c\in Q}$ in sequence:} Broadcast $ID(c)$ and the shortest path distance value $\delta_h (c,c')$ for each $c' \in Q$.	\label{algkSSP:broadcast}
\State \textbf{Local Step at node $\bm{x\in V}$:} For each $c\in Q$ compute the shortest path distance values $\delta (x,c)$ using the distance values received in Step~\ref{algkSSP:broadcast}. \vspace{-0.03in}	\label{localCompute}
\State Run Alg.~\ref{algCase1} and \ref{algCase2} described  in Sec.~\ref{sec:frames} to propagate each distance value $\delta (x,c)$ from source $x \in V$ to blocker node $c \in Q$. \label{communicate}
\State \textbf{For each $\bm{x\in V}$ in sequence:} Compute \textit{extended} $h$-hop shortest paths starting from every $c\in Q$ using Bellman-Ford algorithm (described in Section~\ref{sec:h-hop-extensions}).	\label{extended}
\end{algorithmic}  \label{algAPSP}
\end{algorithm*} 
 
 \noindent
{\bf Derandomizing Distributed Algorithms.}
The method of conditional expectations has been used for derandomizing randomized distributed
algorithms in~\cite{CPS17,GK18}.
Censor-Hillel et al.~\cite{CPS17} semi-formalized a template of combining bounded independence with the method of conditional expectation
for derandomizing an algorithm for
computing a maximal independent set 
in the distributed setting. 
For unweighted graphs, Ghaffari and Kuhn~\cite{GK18} use a special case of the hitting set problem in a bipartite graph  along with a network decomposition 
technique to obtain deterministic distributed
algorithms
 for constructing certain types of spanners, small dominating sets, etc.

Instead of using the method of conditional expectations, our blocker set algorithm 
in Section~\ref{sec:deterministic-blocker}  first gives an efficient
distributed randomized algorithm  for the problem  which uses only pairwise independence.
It uses a linear-sized sample space for generating pairwise independent random 
variables and then
an aggregation of suitable parameters of sample point values
to derandomize our
randomized blocker set algorithm.

\subparagraph*{\bf Roadmap.} In Section~\ref{sec:overall} we present our overall APSP algorithm. Section~\ref{sec:blocker} sketches our blocker set algorithm and Section~\ref{sec:frames} gives our pipelined algorithm for the reversed $q$-sink shortest path problem. Further details on all of our results are
in Appendix~\ref{sec:appendix}.


\subsection{Congest Model}	\label{sec:congest}

In the {\sc Congest} model,
there are $n$ independent processors 
interconnected in a network by bounded-bandwidth links.
We refer to these processors as nodes and
the links as edges.
This network is modeled by graph $G = (V,E)$
where $V$ refer to the set of processors and
$E$ refer to the set of links between the processors.
Here $|V| = n$ and $|E| = m$.

Each node is assigned a unique ID  
between 1 and $poly(n)$ and
has infinite computational power.
Each node has limited topological knowledge 
and only knows about its incident edges.
For the weighted APSP problem we consider, 
 each edge 
has an arbitrary real weight.
Also if the edges are directed, 
the corresponding communication channels are bidirectional
and hence the communication network can be represented 
by the underlying undirected graph $U_G$ of $G$
(as 
in~\cite{HNS17,PR18,GL18}).

The computation proceeds in rounds. In each round each processor can send 
a constant number of 
 words along each outgoing edge,  and it receives the messages sent to it in the previous
round.  The {\sc Congest} model normally assumes that a word has $O(\log n)$ bits.
 Since we allow arbitrary edge-weights, here we
 assume that a constant number of 
  node ids, edge-weights,  and distance values
 can be sent along every edge in every round
(similar assumptions are made 
in~\cite{BN19, ARKP18, Elkin17}).
 The model allows a node to send different message along different edges though we do not
need this feature in our algorithm.

The performance of an algorithm in the \congest{} model is measured by its
round complexity, which is the worst-case
number of rounds of distributed communication.
As noted earlier, 
for the APSP problem, each node in the network needs to compute its shortest path distance from 
every other node
as well as the last edge on each such shortest path.

\section{Overall APSP Algorithm}	\label{sec:overall}

Algorithm~\ref{algAPSP} gives our overall APSP algorithm.
In Step~\ref{computeCSSSP} we use the (simple)
$O(n \cdot h)$-round algorithm in~\cite{AR19} 
to compute 
an $h$-hop \textit{Consistent} SSSP Collection
(or $h$-CSSSP for short)
for the vertex set $V$, defined as follows (and described in detail in Section~\ref{sec:csssp}).
Here,  $\delta(u,v)$ denotes the shortest path distance from $u$ to $v$
and $\delta_h (u,v)$ denotes the $h$-hop shortest path distance from $u$ to $v$.

\begin{definition}[{\bf CSSSP}~\cite{AR19}]	\label{def:CSSSP}
  Let $H$ be a collection of rooted trees of height $h$
  for a set of sources $S \subseteq V$
   in a graph $G=(V,E)$. Then $H$ is an \emph{$h$-hop CSSSP collection} (or simply an \emph{$h$-CSSSP})
  if for every $u, v \in V$ the path from $u$ to $v$ is 
  the same in each of the trees in $H$ (in which such a path exists), and is the $h$-hop shortest path from $u$ to $v$ in the $h$-hop tree $T_u$ rooted at $u$. 
  Further, each $T_u$ contains every vertex $v$ that has a path with at most $h$ hops from $u$ in $G$ that has distance $\delta(u,v)$,
  \end{definition}
  
  The advantage of using 
  $h$-CSSSP 
    instead of 
  other types of
  $h$-hop shortest paths is that 
  the trees in an $h$-CSSSP
  create a consistent collection of paths across all trees in the collection, i.e. a path from $u$ to $v$ is 
same in all trees 
in the
CSSSP collection $\mathcal{C}$ (in which such a path exists).
We exploit this useful property of CSSSPs throughout this paper.

Step~\ref{constructQ} computes a blocker set $Q$, which is defined as follows:

  \begin{definition}[{\bf Blocker Set}~\cite{King99,ARKP18}]
  Let $\Tau$ be a collection of rooted $h$-hop trees
  for a set of vertices $S$
    in a graph $G=(V,E)$. A set $Q\subseteq V$ is a 
\emph{  \centers } for $\Tau$ if every root to leaf path 
 of length $h$ 
  in every tree in $\Tau$ contains a node in $Q$.
  Each node in $Q$ is called a \emph{blocker node} for $\Tau$.
  \end{definition}

Our deterministic blocker set algorithm for Step 2 is completely different from the blocker set algorithms in~\cite{ARKP18, AR19} with 
significant improvement in the round complexity.
We describe this algorithm in Section~\ref{sec:blocker}.
Our blocker set algorithm is based on the
NC approximate
Set Cover algorithm of Berger et al.~\cite{BRS94} and runs in 
$\tilde{O}(|S|\cdot h)$ rounds, where $S$ is the set of
 vertices from which we want to compute the shortest paths.
Previous 
deterministic
blocker set algorithms in~\cite{ARKP18, AR19} have an additional $\tilde{O}(n\cdot |Q|)$ term in the round complexity.

In Step~\ref{computeInCSSSP} we compute,
for each $c\in Q$,
  the $h$-hop in-SSSP rooted at $c$, which is the set of in-coming $h$-hop shortest paths ending at node $c$.
We can compute these $h$-hop in-SSSPs in $O(h)$ rounds per source 
using Bellman-Ford algorithm~\cite{Bellman58}.
In Step~\ref{algkSSP:broadcast} every blocker node $c \in Q$ broadcasts its ID and the corresponding $h$-hop shortest path distance values $\delta_h (c,c')$
for every $c' \in Q$.
Step~\ref{localCompute} is a local computation step where every node $x$ computes its shortest path distances
$\delta(x,c)$
  to every $c \in Q$ using the shortest path distance
values it computed and received in Steps~\ref{computeInCSSSP} and \ref{algkSSP:broadcast} respectively.

In Step~\ref{communicate} every node $x$ wants to send 
each shortest path distance value $\delta (x,c)$ it computed in Step~\ref{localCompute} to blocker node $c\in Q$.
This is the reversed $q$-sink shortest path problem, where $q=|Q|$, and is the other 
crucial
step in our APSP algorithm.
This step
requires sending $\tilde{O}(n^{5/3})$
 different distance values
to
$\tilde{O}(n^{2/3})$
different 
blocker nodes
(using $|Q| = \tilde{O}(n^{2/3})$).
A trivial solution is to broadcast all these messages in the network, resulting in a round complexity of $\tilde{O}(n^{5/3})$  rounds.
However this is the only method known so far to implement this step deterministically.
In Sec.~\ref{sec:frames} we give a pipelined
 algorithm for implementing this step more efficiently in $\tilde{O}(n^{4/3})$ rounds.
After the execution of Step~\ref{communicate} every blocker node $c \in Q$ knows its shortest path distance from every node $x \in V$.

Finally, in Step~\ref{extended}
for every 
$x \in V$, we run  Bellman-Ford algorithm for $h$ hops 
 with distance values $\delta (x,c)$ used as the 
initialization values at every blocker node $c \in Q$.
These constructed paths are also known as 
\textit{extended} $h$-hop shortest 
paths~\cite{HNS17}.
After this step, each $t \in V$ knows the shortest path distance value $\delta (x,t)$ from every 
$x \in V$, which gives the desired APSP output.
We describe Step~\ref{extended} in Section~\ref{sec:h-hop-extensions}.
With these results in place we can now prove Theorem~\ref{thm:main}, whose statement we reproduce here for
convenience.

\begin{theorem*}{\bf \ref{thm:main}}
There is a deterministic distributed algorithm that computes APSP on an $n$-node graph with arbitrary nonnegative 
edge-weights, directed or undirected, in $\tilde{O}(n^{4/3})$ rounds.
\end{theorem*}
\begin{proof}
Fix a pair of nodes $x$ and $t$.
If the shortest path from $x$ to $t$ has less than $h$ hops, then $\delta (x,t) = \delta_h (x,t)$ and the correctness is straightforward 
(see Lemma~\ref{lemma:CSSSP}). 

Otherwise, we can divide the shortest path from $x$ to $t$ into subpaths $x$ to $c_1$, $c_1$ to $c_2$, $\ldots$, $c_l$ to $t$ where 
$c_i \in Q$ for $1\leq i\leq l$ and each of these subpaths have hop-length at most $h$.
Since $x$ knows $\delta_h (x,c_1)$ from Step~\ref{computeInCSSSP} and $\delta_h (c_i,c_{i+1})$ distance values from Step~\ref{algkSSP:broadcast},
it can correctly compute $\delta (x,c_l)$ distance value in Step~\ref{localCompute}.
And from Lemmas~\ref{lemma:algCase1} and \ref{lemma:algCase2}, $c_l$ knows the distance value $\delta (x,c_l)$ after Step~\ref{communicate}.
Since the shortest path from $c_l$ to $t$ has hop-length at most $h$, from Lemma~\ref{lemma:h-hop-extension} $t$ will compute $\delta (x,t)$ in
Step~\ref{extended}.

Step~\ref{computeCSSSP} runs in $O(nh) = O(n^{4/3})$ rounds~\cite{AR19} (Lemma~\ref{lemma:CSSSP}).
In Section~\ref{sec:blocker}, we will give an $\tilde{O}(nh) = \tilde{O}(n^{4/3})$ rounds algorithm to compute a blocker set of size 
$q = \tilde{O}(n/h) =  \tilde{O}(n^{2/3})$ (Step~\ref{constructQ}).
Step~\ref{computeInCSSSP} takes $O(|Q|\cdot h) = \tilde{O}(n)$ rounds using Bellman-Ford algorithm (Lemma~\ref{lemma:CSSSP}).
Since $|Q|^2 = \tilde{O}(n^2 / h^2) = \tilde{O}(n^{4/3})$, Step~\ref{algkSSP:broadcast} takes $\tilde{O}(n^{4/3})$ rounds 
(see Lemma~\ref{lem:all-to-all-bc}).
Step~\ref{localCompute} is local computation and has no communication.
From Lemmas~\ref{lemma:algCase1} and \ref{lemma:algCase2-runtime}, Step~\ref{communicate} takes $\tilde{O}(n^{4/3})$ rounds and Step~\ref{extended} can be computed in $O(nh) = O(n^{4/3})$ rounds
using Lemma~\ref{lemma:h-hop-extension}.
Hence the overall algorithm runs in 
$\tilde{O}(n^{4/3})$ rounds.
\end{proof}

\section{Computing a Blocker Set}	\label{sec:blocker}

\begin{algorithm*}
\scriptsize
\caption{Randomized Blocker Set Algorithm}
Input: $S$: set of source nodes; $h$: number of hops; $\mathcal{C}$: collection of $h$-CSSSP for set $S$; $\epsilon, \delta$: positive constants $\leq 1/12$
\begin{algorithmic}[1]
\State  Compute $score (v)$ for all nodes $v \in V$ using an algorithm from~\cite{ARKP18}. \label{computeScoreV}
\For{{\bf stage $i= \log_{1+\epsilon} n^2$ down to $1$}}	\Comment{All nodes have score value less than $(1+\epsilon)^i$}	\label{startFor1}
	\State Compute $V_i$ and broadcast it using the algorithm described in 
Sec.~\ref{sec:helper}. \label{Vi}
	\State Compute $P^{v}_{i}$ (at each $v \in V$) using 
Algorithm~\ref{algPi}. \label{Pi}
	\For{{\bf phase $j = \log_{1+\epsilon} h$ down to $1$}	} 	\Comment{All paths in $P_i$ have no more than $(1+\epsilon)^j$ nodes in $V_i$}	\label{startFor2}
		\While{there is a path in $P_i$ with at least $(1+\epsilon)^{j-1}$ nodes in $V_i$}	\label{startWhile}
			\State Compute (a) $P^v_{ij}$ (at each $v \in V$) using Alg.~\ref{algPij} and (b) $|P_{ij}|$ using Alg.~\ref{alg|Pij|}. \label{Pij}\label{algRandBlocker:|Pij|}
			\State Compute $score^{ij} (v)$ for all nodes $v \in V_i$ (using a result from~\cite{ARKP18}) and broadcast $score^{ij} (v)$ values. \label{computeScoreij}
			\If{there exists $c \in V_i$ such that $score^{ij} (c) > (\delta^3/(1+\epsilon)) \cdot |P_{ij}|$}	\label{if}
				\State \hspace{.05in} \textbf{Local Step at $\bm{v \in V}$:} add $c$ to $Q$. Break ties with $score^{ij}$ value and node ID.\label{algRandBlocker:pick1}
			\Else		\Comment{Run a selection procedure to select a set of nodes}
				\State  \hspace{.05in} \textbf{Local Step at $\bm{v \in V_i}$:} add $v$ to set $A$ with probability $p = \delta/(1+\epsilon)^j$ (pairwise independently) \footnote{We assume that we have fixed a family of 2-independent hash functions, and every node knows this family. A node $v$ then adds itself to $A$ if $h(v)= 1$.}.	\label{pick2}
				\State \hspace{.05in} \textbf{For each $\bm{v \in V_i}$:} node $v$ broadcast $ID(v)$ if it added itself to $A$ in previous step.	\label{algRandBlocker:broadcastA}
			 	\State \hspace{.05in} \textbf{Local Step at $\bm{v \in V}$:} Check if $A$ is a good set and if so, add $A$ to $Q$. Otherwise, go back to Step~\ref{pick2}. 	\label{algRandBlocker:checkA}
			 \EndIf
			\State \textbf{For each $\bm{x \in S}$ in sequence:} Remove subtrees rooted at $c' \in Q$ using 
Alg.~\ref{algRemoveSubtrees}. \label{algRandBlocker:remove-Subtrees}
			\State Re-compute $score (v)$ for all nodes $v$ and re-construct sets $V_i$ and $P_i$ as described in Steps~\ref{Vi} and \ref{Pi}. 	\label{reconstruct}
		\EndWhile	\label{endWhile}
	\EndFor	\label{endFor2}
\EndFor	\label{endFor1}
\end{algorithmic}  \label{algRandBlocker}
\end{algorithm*}

In this 
section 
we describe our algorithm to compute a small blocker set.
We frame this problem as that of finding a small set cover in an associated hypergraph.
We then adapt the efficient NC algorithm for finding a provably good approximation to this NP-hard problem 
given in Berger et al.~\cite{BRS94}
 to obtain our deterministic distributed algorithm.
 
As in~\cite{BRS94} our algorithm has two parts. We first present a randomized algorithm to find a blocker set of size
 $\tilde{O}(n/h)$
 in $\tilde{O}(|S| \cdot h)$ rounds
using only pairwise independence
\footnote{We use pairwise independence extensively in our analysis of the randomized blocker set algorithm, specifically lemmas A.16, A.17 where we use pairwise independence to get bounds for the terms of the form
$E[X_v \cdot X_{v'}]$.
 This analysis needs pairwise independence as does the derandomization algorithm in Section 3.2.}.
This is described in Section~\ref{sec:rand-blocker}.
Then in Section~\ref{sec:deterministic-blocker} we describe
how to use the exhaustive search technique of Luby~\cite{Luby93} along with the ideas from~\cite{BRS94} to derandomize this 
algorithm, again in  $\tilde{O}(|S| \cdot h)$ rounds. In our overall APSP algorithm $S=V$ but we will also use this algorithm in
Section~\ref{sec:frames} with a different set for $S$.


\subsection{Randomized Blocker Set Algorithm}	\label{sec:rand-blocker}

Given a hypergraph $H=(V,F)$, a subset of vertices $R$ is a set cover for $H$ if $R$ contains at least one vertex in every hyperedge
 in $F$.
Computing a set cover of minimum size is NP-hard. Berger et al.~\cite{BRS94} gave an efficient NC algorithm to compute an $O(\log n)$ approximation to the
minimum set cover.

We now briefly describe the set cover algorithm of Berger et al.~\cite{BRS94}.
The algorithm runs in phases, which are further subdivided into subphases,
in order to construct a suitable blocker set $Q$.
In phase $i$ only vertices 
with degree between $(1+\epsilon)^{i-1}$ and $(1+\epsilon)^i$ are considered for selection 
(let $V_i$ be this set of vertices),
and subphase $j$ consists of only those hyperedges that have at least one vertex in $V_i$.
In each subphase $j$, the algorithm performs a series of selection steps such that all the hyperedges considered in
subphase $j$ are covered by the 
vertices added to $Q$.
This process is repeated for all phases and their subphases and
the set cover is then constructed by taking the union of all the vertices selected by these selection 
steps across all phases.

We 
map the problem of computing a minimum blocker set for an $h$-CSSSP collection $\mathcal{C}$ in a graph $G=(V,E)$ to the minimum set cover problem in 
a
hypergraph $H=(V,F)$ 
as follows. The vertex set
$V$ remains the vertex set of $G$ and each edge in $F$ consists of the vertices in a root-to-leaf path in a tree in $\mathcal{C}$. 
This hypergraph has $n$ vertices and 
at most
$n \cdot |S|$ edges, where $S$ is the number of sources (i.e., trees) in $\mathcal{C}$. Each edge in $F$ has exactly $h$ vertices
(since we do not need to cover paths that have less than $h$ hops).
We now use this mapping to rephrase the algorithm in~\cite{BRS94} in our setting, and we derive an $\tilde{O}(|S| \cdot h)$-round randomized algorithm to compute a
blocker set of 
expected
size within a
$O(\log n)$ factor of the optimal size, using only pairwise independent random variables. 
Since we know there exists a blocker set of size $O((n/h) \cdot \log n)$ (which is constructed in~\cite{ARKP18,AR19})
the size of the blocker set constructed by this randomized algorithm is $\tilde{O}(n/h)$.

Our randomized blocker set method is in Algorithm~\ref{algRandBlocker}.
Table~\ref{table-blocker} presents the notation we use for this section.
In Step~\ref{computeScoreV} for each node $v$ we compute $score(v)$, the number of $h$-hop shortest paths in CSSSP collection $\mathcal{C}$ that contain node $v$. 
This can be done in $O(|S|\cdot h)$ rounds for all nodes $v \in V$ using Algorithm 3 in~\cite{ARKP18}.
Our algorithm proceeds in stages from $i = \log_{1+\epsilon} n^2$ down to $2$ (Steps~\ref{startFor1}-\ref{endFor1}),
where $\epsilon$ is a small positive constant $\leq 1/12$,
such that  at the start of stage $i$, all nodes in $V$ have score value at most  $(1+\epsilon)^i$ and
in stage $i$ we focus on 
$V_i$, the set of  nodes $v$ with $score$ value greater than $(1+\epsilon)^{i-1}$.
(This ensures that the nodes that are added to the blocker set have their score values near the maximum score value).
Let $P_i$ be the set of paths in $\mathcal{C}$ that contain a vertex in 
$V_i$
 and let $P^{v}_{i}$ be the set of paths in $P_i$ with $v$ as the leaf node.
These sets are readily computed in $O(|S| \cdot h)$-rounds (Sec.~\ref{sec:helper}).

\begin{table}
\scriptsize
\centering
\noindent
\caption{Notations} \label{table-blocker}
\begin{tabular}{| c | l |}
\hline
$\mathcal{C}$ & $h$-CSSSP collection	\\
\hline
$S$ & set of source nodes in $\mathcal{C}$\\
\hline
$h$ & number of hops in a path\\
\hline
$n$ & number of nodes \\
\hline
$\epsilon, \delta$ & positive constants $\leq 1/12$	\\
\hline
$Q$ & blocker set (being constructed)\\
\hline
$score(v)$ &  number of root-to-leaf paths in $\mathcal{C}$ that contain $v$ (local var. at $v$)	\\
\hline
$V_i$ & set of nodes $v$ with $score(v) \geq (1+\epsilon)^{i-1}$ \\
\hline
$P_i$ & set of paths in $\mathcal{C}$ with at least one node in $V_i$ \\	
\hline
$P_{ij}$ & set of paths in $P_i$ with at least $(1+\epsilon)^{j-1}$ nodes in $V_i$ \\
\hline
$P^v_i$ & set of paths in $P_i$ with $v$ as the leaf node \\
\hline
$P^v_{ij}$ & set of paths in $P_{ij}$ with $v$ as the leaf node \\
\hline
$score^{ij} (v)$ &  number of paths in $P_{ij}$ that contain $v$ (local var. at $v$)		\\
\hline
\end{tabular}
\end{table}

Similar to~\cite{BRS94}, in order
to ensure that the average number of paths covered by the newly chosen blocker nodes is near the maximum score value, we further divide our algorithm for stage 
$i$ into a sequence of $\log_{1+\epsilon} h = \log_{1+\epsilon} n^{1/3}$ phases, where in each phase $j$ we focus on the paths in $P_i$ with at least 
$(1+\epsilon)^{j-1}$ nodes in $V_i$.
We call this set of paths $P_{ij}$
and let $P^v_{ij}$ be the set of paths in $P_{ij}$ with $v$ as the leaf node.
We maintain that at the start of phase $j$, every path in $P_i$ has at most $(1+\epsilon)^j$ nodes in $V_i$.
We now describe our algorithm for phase $j$ (Steps~\ref{startFor2}-\ref{endFor2}).
The algorithm for phase $j$ consists of a series of selection steps (Steps~\ref{startWhile}-\ref{endWhile}) (similar to~\cite{BRS94})
which are performed until there are no more paths in $P_{ij}$.

Now we describe how we select nodes to add to blocker set $Q$. 
Let $\delta$ be some fixed positive constant less than or equal to $1/12$.
n Step~\ref{if} we check if there exists a node $v$ which covers at least $\delta^3/(1+\epsilon)$ fraction of paths in $P_{ij}$ and if so, we add this 
node to the blocker set in Step~\ref{algRandBlocker:pick1}.
In case of multiple such nodes, we pick the one with the maximum $score^{ij}$ value and break ties using node IDs.
Otherwise in Step~\ref{pick2}, we randomly pick every node with probability $\delta/(1+\epsilon)^j$, pairwise independently, and form a set $A$.
In Step~\ref{algRandBlocker:checkA} we check if $A$ is a good set, otherwise we try again and form a new set $A$ in Step~\ref{pick2}.
As in~\cite{BRS94} we define the notion of a good set as given below and 
we will later show that 
$A$ is a good set with probability at least $1/8$.

\begin{definition}\label{def:good-set}
A set of nodes $A \subseteq V_i$ is a \textit{good set} if $A$ covers at least $|A|\cdot (1+\epsilon)^i \cdot (1-3\delta-\epsilon)$ paths in $P_i$ and at least a $\delta/2$ fraction of paths from $P_{ij}$.
\end{definition}

Before the next selection step, we remove the paths covered by these newly chosen nodes from the collection $\mathcal{C}$ along with
recomputing the score values and sets $V_i$ and $P_i$ (Steps~\ref{algRandBlocker:remove-Subtrees}-\ref{reconstruct}).

\subsubsection{\bf Helper Algorithms for Randomized Blocker Set Algorithm}	\label{sec:helper}

Here we describe the helper algorithms for Algorithm~\ref{algRandBlocker}.

\subparagraph{\bf Algorithms for Computing $V_i$ and $P_i$.}	\label{sec:ViPi}
Here we describe our algorithm for computing Steps 3 and 4 of Algorithm~\ref{algRandBlocker}, which computes the set $V_i$ and 
identifies which paths belong to $P_i$ respectively.
Since every node with score value greater than or equal to $(1+\epsilon)^{i-1}$ belongs to $V_i$, computing $V_i$ is quite trivial.
And to determine if a path $p$ belong to $P_i$, we only need to check if one of the nodes in $p$ is in $V_i$.

Our algorithm for computing $V_i$ works as follows: 
Every node $v$ checks if its score value is greater than or equal to $(1+\epsilon)^{i-1}$ and if so, it broadcast its ID to every other node.
The set $V_i$ is then constructed by including the IDs of all such nodes.
Since there are at most $n$ messages involved in the broadcast step, this algorithm takes $O(n)$ rounds.
This leads to the following lemma.

\begin{lemma}	\label{lemma:V_i}
Given the $score (v)$ values for every $v \in V$, the set $V_i$ can be constructed in $O(n)$ rounds.
\end{lemma} 

We now describe our algorithm for computing $P_i$.
Fix a source node $x \in V$.
In Round $0$ $x$ initializes $flag$ to $true$ if it belongs to $V_i$, otherwise set it to $false$ (Step~\ref{algPi:Step1}).
It then sends this $flag$  to its children in next round (Step~\ref{algPi:send}).
In round $r \geq 1$, a node $v$ that is $r$ hops away from $x$ receives the $flag$ from its parent (Steps~\ref{algPi:receive-Start}-\ref{algPi:receive-End})
and $v$ updates the $flag$ value in Step~\ref{algPi:update} (set it to true if $v \in V_i$) and send it to its children in $x$'s tree 
in round $r+1$ (Step~\ref{algPi:send}).

\begin{algorithm}
\caption{{\sc Compute-$P_i$:} Algorithm for computing paths in $P_i$ for source $x$ at node $v$}
\begin{algorithmic}[1]
\Statex Input: $V_i$; $h$: number of hops; $T_x$: tree for source $x$
\State \textbf{(Round $0$):} \textbf{if $v \in V_i$} \textbf{then} set $flag \leftarrow true$ \textbf{else} $flag \leftarrow false$		\label{algPi:Step1}
\State \textbf{Round $h\geq r > 0$:}
\State \textbf{Send:} \textbf{if} $r = h_x (v) + 1$ \textbf{then} send $\langle flag \rangle$ to all children \vspace{.02in}	\label{algPi:send}
\State {\bf receive [lines~\ref{algPi:receive-Start}-\ref{algPi:receive-End}]:}  
\If{$r = h_x (v)$}	\label{algPi:receive-Start}
	\State let $M$ be the incoming message to $v$ 
	\State let the sender be $w$ and let $M = \langle flag_w \rangle$ and 
	 \State {\bf if} $w$ is a parent of $v$ in $T_x$ {\bf then} $flag \leftarrow flag \vee flag_w$ \label{algPi:update}
\EndIf	\label{algPi:receive-End}
\State \textbf{Local Step at $v$:} \textbf{if} $v$ is a leaf node \textbf{and} $flag = true$ \textbf{then} the path from $x$ to $v$ is in $P_{i}$.
\end{algorithmic}  \label{algPi}
\end{algorithm}

\begin{lemma}	\label{lemma:compute-Pi}
Using  {\sc Compute-$P_i$} (Algorithm~\ref{algPi}), $P_{i}$ can be computed in $O(h)$ rounds per source node.
\end{lemma}

\begin{proof}
Fix a path $p$ from source $x$ to leaf node $v$.
After $h$ rounds, $v$ will know if any node in $p$ belongs to $V_i$ (using the $flag$ value it received in Steps~\ref{algPij:receive-Start}-\ref{algPij:receive-End}).

The algorithm takes $h$ rounds per source $x$ and thus $P_{i}$ can be computed in $O(|S|\cdot h)$ rounds in total (since we need to run the
algorithm for every source $x$).
\end{proof}

\subparagraph{\bf Algorithm for Computing $P_{ij}$.}	\label{sec:Pij}
Here we describe our algorithm for computing Step 7(a) of Algorithm~\ref{algRandBlocker}, which identifies the paths in $P_{i}$ that 
also belong to $P_{ij}$.
Since every path in $P_{ij}$ has at least $(1+\epsilon)^{j-1}$ nodes from $V_i$, for each path $p$ we need to 
determine the number of nodes in $p$ that belong to $V_i$.
We do this by counting the number of nodes that are in $V_i$, starting from root to leaf node.

Our algorithm works as follows:
Fix a source node $x \in V$.
In Round $0$ $x$ initializes $\beta$ value to $1$ if it belongs to $V_i$, otherwise set it to $0$ (Step~\ref{algPij:Step1}).
It then sends this $\beta$ value to its children in next round (Step~\ref{algPij:send}).
In round $r \geq 1$, a node $v$ that is $r$ hops away from $x$ receives the $\beta$ value from its parent (Steps~\ref{algPij:receive-Start}-\ref{algPij:receive-End})
and $v$ updates the $\beta$ value in Step~\ref{algPij:update} (increment it by $1$ if $v \in V_i$) and send it to its children in $x$'s tree 
in round $r+1$ (Step~\ref{algPij:send}).

\begin{algorithm}
\caption{{\sc Compute-$P_{ij}$:} Algorithm for computing paths in $P_{ij}$ for source $x$ at node $v$}
\begin{algorithmic}[1]
\Statex Input: $V_i$; $h$: number of hops; $T_x$: tree for source $x$
\State \textbf{(Round $0$):} \textbf{if $v \in V_i$} set $\beta \leftarrow 1$ \textbf{else} $\beta \leftarrow 0$		\label{algPij:Step1}
\State \textbf{Round $h\geq r > 0$:}
\State \textbf{Send:} \textbf{if} $r = h_x (v) + 1$ \textbf{then} send $\langle \beta \rangle$ to all children \vspace{.02in}	\label{algPij:send}
\State {\bf receive [lines~\ref{algPij:receive-Start}-\ref{algPij:receive-End}]:}  
\If{$r = h_x (v)$}	\label{algPij:receive-Start}
	\State let $\mathcal{M}$ be the incoming message to $v$ 
	\State let the sender be $w$ and let $M = \langle \beta_w \rangle$ and 
	\State {\bf if} $w$ is a parent of $v$ in $T_x$ {\bf then} $\beta \leftarrow \beta + \beta_w$ 	\label{algPij:update}
\EndIf	\label{algPij:receive-End}
\State \textbf{Local Step at $v$:} \textbf{if} $v$ is a leaf node \textbf{and} $\beta \geq (1+\epsilon)^{j-1}$ \textbf{then} the path from $x$ to $v$ is in $P_{ij}$.
\end{algorithmic}  \label{algPij}
\end{algorithm} 

\begin{lemma}	\label{lemma:compute-Pij}
Using  {\sc Compute-$P_{ij}$} (Algorithm~\ref{algPij}), $P_{ij}$ can be computed in $O(h)$ rounds per source node.
\end{lemma}

\begin{proof}
Fix a path $p$ from source $x$ to leaf node $v$.
After $h$ rounds, $v$ will know the number of nodes that belong to $V_i$ (using the $\beta$ values it received in Steps~\ref{algPij:receive-Start}-\ref{algPij:receive-End}).

The algorithm takes $h$ rounds per source $x$ and thus $P_{ij}$ can be computed in $O(|S|\cdot h)$ rounds in total (since we need to run the
algorithm for every source $x$).
\end{proof}

\subparagraph{\bf Algorithm for Computing $|P_{ij}|$.}	\label{sec:|Pij|}
Algorithm~\ref{alg|Pij|} describes our algorithm for computing 
Step 7(b) of Algorithm~\ref{algRandBlocker}, which computes the value of $|P_{ij}|$.
Let $P^{v}_{ij}$ represents the set of paths $p$ in $P_{ij}$ with $v$ as the leaf node.
Every node $v$ knows the set $P^{v}_{ij}$ after running the algorithm described in the previous section.
Our algorithm works as follows:
Every node $v$ first compute $|P^v_{ij}|$ (Step~\ref{alg|Pij|:Step1}) and then broadcast this value in Step~\ref{alg|Pij|:broadcast}.
Every node $v$ then compute $|P_{ij}|$ by summing up the values received in Step~\ref{alg|Pij|:broadcast} (Step~\ref{alg|Pij|:local}).

\begin{algorithm}
\caption{{\sc Compute-$|P_{ij}|$}}
\begin{algorithmic}[1]
\Statex Input: $P^v_{ij}$: paths in $P_{ij}$ with $v$ as the leaf node
\State \textbf{Local Step at $v \in V$:} set $\alpha_{P^v_{ij}} \leftarrow |P^v_{ij}|$	\label{alg|Pij|:Step1}
\State \textbf{For each $v \in V$:} Broadcast $ID (v)$ and the value $\alpha_{P^v_{ij}}$. \label{alg|Pij|:broadcast}
\State \textbf{Local Step at $v \in V$:} $|P_{ij}| \leftarrow \sum_{v' \in V} \alpha_{P^{v'}_{ij}}$	\label{alg|Pij|:local}
\end{algorithmic}  \label{alg|Pij|}
\end{algorithm} 

\begin{lemma}	\label{lemma:|Pij|}
{\sc Compute-$|P_{ij}|$} (Algorithm~\ref{alg|Pij|}) computes $|P_{ij}|$ in $O(n)$ rounds.
\end{lemma}

\begin{proof}
Steps~\ref{alg|Pij|:Step1} and \ref{alg|Pij|:local} are local steps and involves no communication.
Step~\ref{alg|Pij|:broadcast} involves a broadcast of $n$ messages and takes $O(n)$ rounds using Lemma~\ref{lem:all-to-all-bc}.
\end{proof}

\subparagraph{\bf Remove Subtrees rooted at $z \in Z$.}	\label{sec:remove-subtree}
In this Section we describe a deterministic algorithm for 
implementing Step 15 of Algorithm~\ref{algRandBlocker}, which
removes subtrees rooted at nodes $z \in Z$ from the trees in the given $h$-CSSSP collection $\mathcal{C}$.
This algorithm (Algorithm~\ref{algRemoveSubtrees}) is quite simple and works as follows: 
Fix a source $x$ and let its corresponding tree in $\mathcal{C}$ be $T_x$.
Every node $z \in Z$ in $T_x$ send its ID to all its children in $T_x$ (Step~\ref{algRemoveSubtrees:Step1}).
Every node $v$ on receiving a message from its parent in $T_x$, forwards it to all its children and set the parent pointer in $T_x$ to $NIL$
(Step~\ref{algRemoveSubtrees:Step2}).

\begin{algorithm}
\caption{{\sc Remove-Subtrees:} Algorithm for Removing Subtrees rooted at $z \in Z$ for source $x$ at node $v$}
Input: $S$: set of sources; $\mathcal{C}$: $h$-CSSSP collection for set $S$
\begin{algorithmic}[1]
\State \textbf{(Round $\bm{0}$:)} \textbf{If} $v \in Z$ \textbf{then} send $\langle ID(v) \rangle$ to all children in $T_x$ and set $parent_x (v)$ to $NIL$. 	\label{algRemoveSubtrees:Step1}
\State \textbf{(Round $r > 0$:)} \textbf{If} $v$ received a message $M$ in round $r-1$ \textbf{then} set $parent_x (v)$ to $NIL$ and send $M$ to all children in $T_x$.	\label{algRemoveSubtrees:Step2}
\end{algorithmic}  \label{algRemoveSubtrees}
\end{algorithm}

\begin{lemma}	\label{lemma:remove}
Given a source $x \in S$ and tree $T_x \in \mathcal{C}$, then {\sc Remove-Subtrees} (Algorithm~\ref{algRemoveSubtrees}) removes all subtrees rooted at $z \in Z$ in $T_x$.
\end{lemma}

\begin{proof}
Every $z \in Z$ in $T_x$ removes its parent pointer in $T_x$ in Step~\ref{algRemoveSubtrees:Step1}.
Any node $v \in V$ that lies in the subtree rooted at a $z \in Z$ in $T_x$ would have received a message with $ID(z)$ from its parent by 
$h$ rounds (since height of $T_x$ is at most $h$) and hence would have set its parent pointer to NIL in Step~\ref{algRemoveSubtrees:Step2}. 
\end{proof}

\begin{lemma}	\label{lemma:runTime}
{\sc Remove-Subtrees} (Algorithm~\ref{algRemoveSubtrees}) requires at most $h$ rounds per source node $x \in S$.
\end{lemma}

\begin{proof}
Since the height of $T_x$ is at most $h$, any node $v \in V$ which lies in the subtree rooted at a $z \in Z$ will receive the message from 
$z$ by $h$ rounds. 
This establishes the lemma.
\end{proof}

\subsubsection{\bf Correctness of Algorithm~\ref{algRandBlocker}}

Similar to~\cite{BRS94} we get the following 
Lemmas~\ref{lemma:A}-\ref{lemma:BlockerSize}
which give us a bound on the number of 
selection steps and a bound on the size of $Q$.

\begin{lemma}\label{lemma:A}
The set $A$ constructed in Step~\ref{pick2} 
is a good set with probability at least $1/8$.
\end{lemma}

\begin{lemma}\label{lemma:selection}
The while loop in Steps~\ref{startWhile}-\ref{endWhile} runs for at most $O\left(\frac{\log^3 n}{\delta^3\cdot \epsilon^2}\right)$ iterations in total.
\end{lemma}

We provide the full proof of Lemmas~\ref{lemma:A} and \ref{lemma:selection} in Sec~\ref{sec:correctness-rand-blocker}.

\begin{lemma}\label{lemma:BlockerSize}
The blocker set $Q$ constructed by Algorithm~\ref{algRandBlocker} has size $O(n\log n / h)$.
\end{lemma}

\begin{proof}
As shown in~\cite{King99, ARKP18} the size of the blocker set computed by an optimal greedy algorithm is
 $\Theta\left(\frac{n\ln p}{h}\right)$, where $p$ is the number of paths 
that need to be covered.
We will now argue that the blocker set constructed by Algorithm~\ref{algRandBlocker} is at most a factor of $\frac{1}{(1- 3\delta -\epsilon)}$ larger than the
greedy solution, thus showing that the constructed blocker set $Q$ has size at most $O\left(\frac{n\ln p}{h}\cdot \frac{1}{(1- 3\delta -\epsilon)}\right) = 
\tilde{O}\left(\frac{n\log n}{h}\right)$ since $p \leq n^2$ and $0 < \delta, \epsilon \leq \frac{1}{12}$.

The blocker set $Q$ constructed by Algorithm~\ref{algRandBlocker} has 2 types of nodes: (1) node $c$ added in Step~\ref{algRandBlocker:pick1}, (2) set of nodes 
$A$ added in Step~\ref{pick2}.
Since the while loop in Steps~\ref{startWhile}-\ref{endWhile} runs for at most $O\left(\frac{\log^3 n}{\delta^3\cdot \epsilon^2} \right)$ iterations (by Lemma~\ref{lemma:selection}), hence there are at most 
$O\left(\frac{\log^3 n}{\delta^3\cdot \epsilon^2} \right)$ nodes of type 1.
Since $\frac{\log^3 n}{\delta^3\cdot \epsilon^2} = o(\frac{n}{h})$, hence we only need to bound the number of nodes added in 
Steps~\ref{pick2}-\ref{algRandBlocker:checkA}.

Since $A$ is a good set, by Lemma~\ref{lemma:A} the number of paths covered by $A$ is at least 
$|A|\cdot (1+\epsilon)^i\cdot (1-3\delta-\epsilon)$, where $(1+\epsilon)^i$ is the maximum possible score value 
across all nodes in $V$ (in the current iteration).
Since maximum possible score value is $(1+\epsilon)^i$, any greedy solution must add at least 
$|A|\cdot (1-3\delta-\epsilon)$ nodes in the blocker set to cover these paths. 
Hence the choice of $A$ is at most a factor of $\frac{1}{(1-3\delta-\epsilon)}$ larger than the greedy solution.
This establishes the lemma.
 \end{proof}

\begin{lemma}	\label{runTime}
Algorithm~\ref{algRandBlocker} computes the blocker set $Q$ in $\tilde{O}(|S|\cdot h/ (\epsilon^2 \delta^3))$ rounds,
in expectation.
\end{lemma}

\begin{proof}
Step~\ref{computeScoreV} runs in $O(|S|\cdot h)$ rounds~\cite{ARKP18}. 
 The for loop in Steps~\ref{startFor1}-\ref{endFor1} runs for $\log_{1+\epsilon} n^2 = O\left(\log n / \epsilon \right)$ 
iterations.
Each iteration 
takes $\tilde{O}\left(|S|\cdot h / (\epsilon \delta^3) \right)$ rounds in
expectation:
Step~\ref{Pi} is readily seen to run in $O(|S|\cdot h)$ rounds 
(Lemma~\ref{lemma:compute-Pi}).
The inner for loop in Steps~\ref{startFor2}-\ref{endFor2} runs for $\log_{1+\epsilon} h = O\left(\log n / \epsilon \right)$
iterations, with each iteration taking $\tilde{O}\left(|S|\cdot h / \delta^3 \right)$ rounds in expectation 
(Lemma A.12).
\end{proof}

\subsection{Deterministic   Blocker Set Algorithm}	\label{sec:deterministic-blocker}

The only place where randomization is used in Algorithm~\ref{algRandBlocker} is in 
Steps~\ref{pick2}-\ref{algRandBlocker:checkA}, where a good set $A$ (see Definition~\ref{def:good-set}) is chosen.
Fortunately, the $X_v$'s are pairwise-independent random variables, where $X_v = 1$ if $v\in A$ and $0$ otherwise.  We
use an $O(n)$ size
sample space~\cite{LW06,Luby93} for generating pairwise independent random variables 
and then 
 find a good sample point(i.e., a good set $A$)  in this $O(n)$-sized sample space in $O(|S|\cdot h + n)$ rounds.
This $O(n)$ size sample space can be generated locally at each node by going over all $0$-$1$ strings of length 
about $\log n$ using the techniques in~\cite{LW06,Luby93}.
We provide more details about the construction of this sample space in Appendix~\ref{sec:sample-space}.

\begin{algorithm*}	
\caption{Deterministic Algorithm for picking good set $A$}
\begin{algorithmic}[1]
\Statex Input: $h$: number of hops; $S$: set of source nodes; $\mathcal{C}$: $h$-CSSSP collection; $X^{(\mu)}$: $\mu$-th vector in sample space 
\State \textbf{For each $x \in S$ in sequence:} Collect  at each $v$ the ids of the nodes on the path ending at leaf node $v$ in $T_x$ (using Alg. 4 in~\cite{ARKP18}). (This computes both $P^v_i$ and $P^v_{ij}$) \label{algDet:Anc} 
\State Compute BFS in-tree $T$ rooted at leader $l$. \label{algDet:init}
\State Compute $\sigma_{P_i, u}^{(\mu)}$ and $\sigma_{P_{ij}, u}^{(\mu)}$ terms locally at each $v \in V$, for each sample point $\mu$, and then using the pipelined algorithm in Sec. A.5, send these values to the leader $l$.	\label{algDet:compute}
\State \textbf{Local Step at $\bm{l}$:} For each $1\leq \mu\leq n$, compute $\nu_{P_i}^{(\mu)}$ and $\nu_{P_{ij}}^{(\mu)}$. Let $\mu'$ be such that $X^{(\mu^{'})}$ corresponds to a good set $A$. \label{algDet:B} 
\State Node $l$ broadcast $X^{(\mu')}$ values. (This corresponds to good set $A$) \label{algDet:broadcast}
\end{algorithmic} \label{algDet}
\end{algorithm*}

Algorithm~\ref{algDet}, our derandomized algorithm, works as follows.
Each vertex $v$ generates the sample set locally.
Let $P_i^v$ and $P_{ij}^v$ denote the set of paths in $P_i$ and $P_{ij}$, respectively, that have $v$ as the leaf node.
Initially every node $v$ determines these sets $P_i^v$ and $P_{ij}^v$, which can be done in $O(|S|\cdot h)$ rounds.
Then for each sample point $\mu$, every node $v$
locally
computes the number of paths in sets $P_i^v$ and $P_{ij}^v$ covered by $\mu$.
Each vertex sends its
computed values
for all sample points
  to a leader node
$l \in V$  which then computes the total number of paths covered in both $P_i$ and $P_{ij}$
for every 
sample point $\mu$
and picks 
one that satisfies the good set criterion (Definition~\ref{def:good-set}).
Such a set is guaranteed to exist from Lemma~\ref{lemma:A}.

To compute sets $P_i^v$ and $P_{ij}^v$, we
collect the ids of the vertices in each $P_i$ at the leaf node of $P_i$, for each tree $T_x$ 
in turn, 
using the {\sc Ancestors} algorithm in~\cite{ARKP18} (Step~\ref{algDet:Anc}, Alg.~\ref{algDet}).
We then create an incoming BFS tree rooted at $l$ (Step~\ref{algDet:init}, Alg.~\ref{algDet}).
We assume that the $X$ values are enumerated in order and every node knows this enumeration.
Let $X^{(\mu)}$ refer to the $\mu$-th vector in this enumeration
and let $\sigma^{(\mu)}_{P_i, v}$ and  $\sigma^{(\mu)}_{P_{ij}, v}$ refer to the number of paths covered by $X^{(\mu)}$
in sets $P_i^v$ and $P_{ij}^v$ respectively.
Similarly let $\nu^{(\mu)}_{P_i}$ and  $\nu^{(\mu)}_{P_{ij}}$ refer to the total number of paths covered by $X^{(\mu)}$
in sets $P_i$ and $P_{ij}$ respectively.
In Step~\ref{algDet:compute} (Alg.~\ref{algDet}), the leader $l$
receive sums of the $\nu_{P_i, u}$ and $\nu_{P_{ij}, u}$ values for all sample points from the nodes $u$
using the 
algorithm in Sec. A.5.
The leader then is able to compute the number of paths covered in both $P_i$ and $P_{ij}$ for each $\mu$ and
 then picks one that satisfies the good set criterion
(Step~\ref{algDet:B}, Alg.~\ref{algDet}).
It then broadcasts the corresponding $X$ vector to every node in the network (Step~\ref{algDet:broadcast}, Alg.~\ref{algDet}).
Algorithm~\ref{algDet} gives the pseudocode for this algorithm.

\begin{lemma}	\label{lemma:deterministic}
The leader node $l$ can identify a good sample point $X \in \{0,1\}^{|V_i|}$, and thus a good set $A$ in $O(|S|\cdot h + n)$ rounds.
\end{lemma}

Details of the steps in  Alg.~\ref{algDet} and proof of Lemma~\ref{lemma:deterministic} are in Sec. A.5.
Let Algorithm $2'$ be the blocker set algorithm obtained after replacing Steps~\ref{pick2}-\ref{algRandBlocker:checkA} in 
Algorithm~\ref{algRandBlocker} with the deterministic algorithm for generating a good set $A$
(Algorithm~\ref{algDet}).
Lemma~\ref{lemma:deterministic}, together with Lemma~\ref{runTime}, gives us the following Corollary.

\begin{corollary}	\label{runTime-deterministic}
Algorithm $2'$ computes the blocker set $Q$ 
deterministically in $\tilde{O}(|S|\cdot h / (\epsilon^2 \delta^3))$ rounds.
\end{corollary}

\section{A $\tilde{O}(n^{4/3})$ Rounds Algorithm for Step~\ref{communicate} of Algorithm~\ref{algAPSP}}	\label{sec:frames}

In Step~\ref{communicate} of Algorithm~\ref{algAPSP}, 
the goal is to send the distance values $\delta (x,c)$ (which are already computed at node $x$)
from source node $x \in V$ to the corresponding blocker node $c$.
Since there are $n$ sources and $|Q| = \tilde{O}(n^{2/3})$ blocker nodes, 
this step can be
 implemented in $\tilde{O}(n^{5/3})$ rounds using all-to-all broadcast (Lemma~\ref{lem:all-to-all-bc}).
One could conjecture that the techniques in~\cite{HPDG+19, LSP19} could be used to send these $\tilde{O}(n^{5/3})$
 messages from the source nodes  in $V$ 
 to the blocker nodes by constructing trees rooted
 at each $c$.
However, 
it is not clear how these methods can distribute the $\tilde{O}(n^{5/3})$ different 
source-destination
messages in $o(n^{5/3})$ rounds.

We now describe a method to implement this step more efficiently in $\tilde{O}(n^{4/3})$ rounds deterministically.
 A randomized  $\tilde{O}(n^{4/3})$-round algorithm for this problem is given in Huang et al.~\cite{HNS17}.
 Our algorithm uses the concept of {\it bottleneck nodes} from that result but is otherwise quite
 different.
 
 Our algorithm is divided into two cases: (i) when $hops (x,c) > n^{2/3}$ and, (ii) when $hops (x,c) \leq n^{2/3}$ ($hops (x,c)$ denotes the number of edges on the shortest path from $x$ to $c$).

\subsection{Case (i): $\bm{hops (x,c) > n^{2/3}}$}

Algorithm~\ref{algCase1} describes our algorithm for this case.
We first construct an $n^{2/3}$-in-CSSSP collection
(i.e., CSSSP in-trees)
using
 the blocker set $Q$ as the source set
(Step~\ref{algCase1:computeCSSSP}, Alg.~\ref{algCase1}).
In Step~\ref{algCase1:constructQ} (Alg.~\ref{algCase1}) we
construct a blocker set $Q'$ of size $\tilde{O}(n^{1/3})$ for this CSSSP collection using 
deterministic Algorithm $2'$
in 
Sections~\ref{sec:blocker} and  \ref{sec:deterministic-blocker}.
Then for each $c' \in Q'$ we construct the incoming and outgoing shortest path tree rooted at $c'$ 
(Step~\ref{algCase1:computeSSSP}, Alg.~\ref{algCase1}).
In Step~\ref{algCase1:broadcast} (Alg.~\ref{algCase1}), every source $x\in V$ broadcasts the 
distance value $\delta (x,c')$ for each $c' \in Q'$.
The lemma below shows that each $c\in Q$ can determine the $\delta(x,c)$ values for all $x$ for which
$hops(x,c) > n^{2/3}$, and the algorithm runs in $\tilde{O}(n^{4/3})$ rounds.

\begin{algorithm*}
\caption{Compute $\delta(x,c)$ at $c$: \textbf{when $\bm{hops(x,c) > n^{2/3}}$}}
\begin{algorithmic}[1]
\Statex Input: $Q$: blocker set
\State Compute $n^{2/3}$-in-CSSSP for source set $Q$ using the algorithm in~\cite{AR19}. \label{algCase1:computeCSSSP}
\State Compute a blocker set $Q'$ of size $\tilde{O}(n / n^{2/3}) = \tilde{O}(n^{1/3})$ for the $n^{2/3}$-CSSSP computed in Step~\ref{computeCSSSP} using the blocker set algorithm described in Section~\ref{sec:blocker}.	\label{algCase1:constructQ}
\State \textbf{For each $\bm{c' \in Q'}$ in sequence:} Compute in-SSSP and out-SSSP rooted at $c'$ using Bellman-Ford algorithm. \label{algCase1:computeSSSP}
\State \textbf{For each $\bm{x \in V}$ in sequence:} Broadcast $ID(x)$ and the shortest path distance values $\delta (x,c')$ for each $c' \in Q'$.	\label{algCase1:broadcast}
\State \textbf{Local Step at node $\bm{c \in Q}$:} For each $x \in V$ compute the shortest path distance value $\delta (x,c)$ using the $\delta
 (x,c')$ distance values received in Step~\ref{algCase1:broadcast} and the $\delta (c', c)$ distance values computed in Step~\ref{algCase1:computeSSSP}.	\label{algCase1:localCompute}
\end{algorithmic}  \label{algCase1}
\end{algorithm*}

\begin{lemma}	\label{lemma:algCase1}
Let $V'$ be the set of nodes $x$ such that there is a shortest path from $x$ to a blocker node $c \in Q$ with  hop-length greater than $n^{2/3}$. Using
Algorithm~\ref{algCase1} each blocker node $c$ can correctly compute $\delta (x,c)$ for all such $x \in V'$ in $\tilde{O}(n^{4/3})$ rounds.
\end{lemma}

\begin{proof}
Since $hops (x,c) > n^{2/3}$, there exists a blocker node $c' \in Q'$ (constructed in Step~\ref{algCase1:constructQ}) such that the shortest path from $x$ to
$c$ passes through $c'$. 
Thus $c$ can compute the distance value $\delta (x,c)$ by adding $\delta (x,c')$ (received in Step~\ref{algCase1:broadcast}) and
$\delta (c',c)$ (computed in Step~\ref{algCase1:computeSSSP}) values in Step~\ref{algCase1:localCompute}.

Step~\ref{algCase1:computeCSSSP} takes $O(n^{2/3}\cdot |Q|) = \tilde{O}(n^{4/3})$ rounds using Bellman-Ford algorithm.
Step~\ref{algCase1:constructQ} requires $\tilde{O}(n^{2/3}\cdot n^{2/3}) = \tilde{O}(n^{4/3})$ rounds by Corollary~\ref{runTime-deterministic}.
Since $|Q'| = \tilde{O}(n/n^{2/3}) = \tilde{O}(n^{1/3})$, Step~\ref{algCase1:computeSSSP} takes 
$\tilde{O}(n\cdot n^{1/3}) = \tilde{O}(n^{4/3})$ rounds using Bellman-Ford algorithm and so does Step~\ref{algCase1:broadcast} 
using Lemma~\ref{lem:all-to-all-bc}.
Step~\ref{algCase1:localCompute} is a local step and has no communication.
\end{proof}

\begin{algorithm*}
\caption{Compute $\delta(x,c)$ at $c$: \textbf{when $\bm{hops(x,c) \leq n^{2/3}}$}}
\begin{algorithmic}[1]
\Statex Input: $Q$: blocker set; $|Q| \leq n^{2/3} \log n$; $\mathcal{C}^{Q}$: $n^{2/3}$-in-CSSSP collection for set $Q$
\State Compute a set of bottleneck nodes $B$ of size $\tilde{O}(n^{1/3})$ using Algorithm 13 (Sec. A.6.1).
 \label{algCase2:bottleneck}
\State \textbf{For each $\bm{b \in B}$ in sequence:} Compute both in-SSSP and out-SSSP tree rooted at $b$ using Bellman-Ford algorithm.	\label{algCase2:compute-SSSP}
\State \textbf{For each $\bm{x \in V}$ in sequence:} Broadcast $ID(x)$ and the shortest path distance values $\delta (x,b)$ for each $b \in B$.	\label{algCase2:broadcast}
\State \textbf{Local Step at node $\bm{c \in Q}$:} For each $x \in V$ compute $\delta^{(B)} (x,c) = \min_{b \in B} \{\delta (x,b) + \delta (b,c) \}$ using the $\delta (x,b)$ distance values received in Step~\ref{algCase2:broadcast} and $\delta (b,c)$ distance values computed in Step~\ref{algCase2:compute-SSSP}.	\label{algCase2:localCompute}
\State Remove subtrees rooted at $b \in B$ from the collection $\mathcal{C}^{Q}$ using Algorithm~\ref{algRemoveSubtrees}. \label{algCase2:remove-Subtrees}
\State Reset round counter to $0$.	\label{algCase2:reset}
\State Assume the nodes in $Q$ are ordered in a (cyclic) sequence $O$. 
\State \textbf{Round $0 < r \leq (n^{4/3} \log n + n^{4/3}) \cdot ((1/3)\cdot \log n / \log \log n - 1)$:}	\label{algCase2:round}
\State \hspace{.05in} {\bf Round-robin sends:} At each node $v$, forward an unsent message for the next blocker node $c$ in $O$ to its parent in $c$'s tree.\label{algCase2:round-robin}
\end{algorithmic}  \label{algCase2}
\end{algorithm*}

\subsection{Case (ii): $\bm{hops (x,c) \leq n^{2/3}}$}
This case deals with sending the distance values from source nodes $x\in V$
to the blocker nodes $c$ when the shortest path between 
$x$ and $c$ has hop-length at most $n^{2/3}$.
Recall that using an all-to-all broadcast or the techniques in~\cite{HPDG+19, LSP19} for sending these $\tilde{O}(n^{5/3})$ messages
appears to
require at least $\tilde{O}(n^{5/3})$ rounds.

 Let $\mathcal{C}^{Q}$ be the $n^{2/3}$-in-CSSSP collection for  source set $Q$.
 A set $B \subset V$ is a set of bottleneck nodes if removing the nodes in $B$, along with their descendants in the trees 
 in the collection $\mathcal{C}^{Q}$, reduces the congestion to at most $\tilde{O}(n^{4/3})$, i.e. every node 
 would need to send at most $\tilde{O}(n^{4/3})$ messages if all nodes $x$ transmitted their $\delta(x,c)$ values along the pruned CSSSP trees in the collection $\mathcal{C}^{Q}$.
 This notion is defined in Huang et al.~\cite{HNS17}, where
 they present a randomized algorithm using 
 the randomized scheduling algorithm in
 Ghaffari~\cite{Ghaffari15} to identify 
such a set of 
 bottleneck nodes.
 Here we 
 deterministically identify a set of
 bottleneck nodes $B$ where $|B| = \tilde{O}(n^{1/3})$ (Step~\ref{algCase2:bottleneck}, Alg.~\ref{algCase2})
 using a pipelined strategy 
 (Sec. A.6.1).
 Clearly, after we remove 
 these bottleneck nodes, 
 any remaining node needs to send at most 
 $\tilde{O}(n^{4/3})$ messages.
 
After we identify the set of bottleneck nodes $B$ we
 run Bellman-Ford algorithm~\cite{Bellman58} for each $b \in B$ to compute both the incoming and outgoing
 shortest path tree rooted at  $b$ (Step~\ref{algCase2:compute-SSSP}, Alg.~\ref{algCase2}).
 We then broadcast the $\delta (x,b)$ distance values from every source $x \in V$ to the corresponding
 $b \in B$ (Step~\ref{algCase2:broadcast}, Alg.~\ref{algCase2}).
Thus for all vertices $x \in V$ such that the shortest path from $x$ to 
a blocker node 
$c$ passes through some 
other
$b \in B$, the 
blocker node $c$ can compute the shortest path distance value, $\delta (x,c)$ by adding $\delta (x,b)$ and
$\delta (b,c)$ distance values (Step~\ref{algCase2:localCompute}, Alg.~\ref{algCase2}).
 
 It remains
   to send the distance value $\delta (x,c)$ to blocker node $c$ if $x$ is not part of a subtree of any bottleneck node 
 $b$ in $c$'s shortest path tree.
 Since the maximum congestion at any node is at most $\tilde{O}(n^{4/3})$ after removing bottleneck nodes in $B$, we are able to perform this computation
 deterministically.
In Steps~\ref{algCase2:round}-\ref{algCase2:round-robin} (Alg.~\ref{algCase2}), we use a simple round-robin
 strategy to propagate these distance values from each source $x \in V$ to all blocker nodes $c$ in the network.
We show in Section~\ref{sec:correctness}, using the notion of \textit{frames},  that this simple strategy
achieves the desired $\tilde{O}(n^{4/3})$-round bound.

\begin{lemma}	\label{lemma:algCase2:Steps1-4}
If the shortest path from $x \in V$ to a blocker node $c \in Q$ has hop-length at most $n^{2/3}$ and there exists a bottleneck node $b \in B$ on this path, then after executing Steps~\ref{algCase2:bottleneck}-\ref{algCase2:localCompute} of Algorithm~\ref{algCase2} blocker node $c$ knows the distance value $\delta (x,c)$ for all such $x \in V$.
\end{lemma}

\begin{proof}
This is immediate from Step~\ref{algCase2:localCompute} (Alg.~\ref{algCase2}) where $c$ will compute $\delta (x,c)$ by adding the distance values 
$\delta (x,b)$ (received in Step~\ref{algCase2:broadcast}, Alg.~\ref{algCase2}) and $\delta (b,c)$ value (computed at $c$ in Step~\ref{algCase2:compute-SSSP}, Alg.~\ref{algCase2}).
\end{proof}

\begin{lemma}	\label{lemma:algCase2:round-robin}
If a source node $x\in V$ lies in a blocker node $c$'s tree in the CSSSP collection $\mathcal{C}^Q$ after the execution of Step~\ref{algCase2:remove-Subtrees} of Algorithm~\ref{algCase2}, then $c$ would have received $\delta (x,c)$ value by $(n^{4/3} \log n + n^{4/3})\cdot ((1/3)\cdot \log n / \log \log n - 1)$ rounds of Step~\ref{algCase2:round-robin} of Algorithm~\ref{algCase2}.
\end{lemma}

Lemma~\ref{lemma:algCase2:round-robin} is established below in Section~\ref{sec:correctness}.
Lemmas~\ref{lemma:algCase2:Steps1-4} and \ref{lemma:algCase2:round-robin}  establish the
following lemma.

\begin{lemma}	\label{lemma:algCase2}
If the shortest path from $x \in V$ to a blocker node $c \in Q$ has hop-length at most $n^{2/3}$, then after running Algorithm~\ref{algCase2} blocker node $c$ knows the distance value $\delta (x,c)$ for all such $x \in V$.
\end{lemma}

\begin{lemma}	\label{lemma:algCase2-runtime}
Algorithm~\ref{algCase2} runs for $\tilde{O}(n^{4/3})$ rounds in total.
\end{lemma}
\begin{proof}
Step~\ref{algCase2:bottleneck} takes $\tilde{O}(n^{4/3})$ rounds (Lemma~\ref{lemma:algBottleneck}). 
Since $|B| = \tilde{O}(n^{1/3})$, Step~\ref{algCase2:compute-SSSP} takes $\tilde{O}(n\cdot n^{1/3}) = \tilde{O}(n^{4/3})$ rounds
using Bellman-Ford algorithm and so does Step~\ref{algCase2:broadcast} using Lemma~\ref{lem:all-to-all-bc}.
Step~\ref{algCase2:localCompute} is a local step and involves no communication.
Step~\ref{algCase2:remove-Subtrees} takes $\tilde{O}(n^{2/3}\cdot |Q|) = \tilde{O}(n^{4/3})$ rounds
 (Lemma~\ref{lemma:runTime}).
Step~\ref{algCase2:round-robin} runs for $\tilde{O}(n^{4/3})$ rounds, thus establishing the lemma.
\end{proof}

\subsection{Correctness of Step~\ref{algCase2:round-robin} of Algorithm~\ref{algCase2}}	\label{sec:correctness}

In this section we will establish that  the simple round-robin approach used in Steps~\ref{algCase2:round}-\ref{algCase2:round-robin}
of Algorithm~\ref{algCase2} is sufficient to propagate distance values $\delta (x,c)$ from source nodes 
$x \in V$ to blocker nodes $c \in Q$ in $\tilde{O}(n^{4/3})$ rounds, 
when the congestion at any node\footnote{Congestion at a node refers to the maximum number of messages sent by a node during the execution of an algorithm.}
 is at most 
$n\sqrt{|Q|}$.
While this looks plausible, the issue to resolve is whether a node could be left idling 
when there are more messages it needs to pass on from its descendants to its parents in some of the trees.
This could happen because each node forwards at most one message per round and these
descendants might have forwarded messages for other blocker nodes. The round robin
scheme appears to only guarantee that a message for a chosen blocker node will be sent from a node to
its parent at least once every $|Q|$ rounds.

We now present and analyze a more structured version of Steps 9-10  to establish the bound.
In this Algorithm 6 we
divide Step~\ref{algCase2:round-robin} (Alg.~\ref{algCase2}) into 
$(1/3)\cdot (\log n/\log \log n) - 1$
different {\it stages}, with each stage running for at most  
$n^{4/3} \log n + n^{4/3}$ rounds  (we assume $|Q| \leq n^{2/3} \log n$).
Our key observation (in  Lemma~\ref{lemma:Qvj}) is that
at the start of Stage $i$, every node $v$ only needs to send the distance values for at most 
$n^{2/3}/\log^{i+1/2} n$ different blocker nodes (note that $i$ is not a constant), 
thus more messages can be sent by $v$ to each blocker node in later stages.

Let $Q_{v,i}$ be the set of blocker nodes for which node $v$ has messages to send at start of stage $i$.
We introduce the notion of a \textit{frame}, where each frame has a single round available for each blocker node in $Q_{v,i}$.
Stage $i$ is divided into $n^{2/3} \log^{i + 1} n + n^{2/3}$ frames 
(we will show that each frame consists of 
$\lceil n^{2/3}/\log^{i+1/2} n \rceil$ rounds).
In each frame, node $v$ sends out an unsent message for each $c \in Q_{v,i}$ to its parent in $c$'s tree (Step~\ref{algStage-i:forward}, Alg.~\ref{algStage-i}).

\begin{algorithm}
\caption{Algorithm for Stage $i$ at node $v \in V$}
\begin{algorithmic}[1]
\Statex Input: $O$: (cyclic) sequence of nodes in blocker set $Q$; $\mathcal{C}^{Q}$: $n^{2/3}$-CSSSP collection for set $Q$
\For{$\bm{i \geq 0}:$}
	\State Let $Q_{v,i}$ be the set of nodes in $Q$ for which $v$ contains at least one unsent message during Stage-$i$. \label{algStage-i:localCompute}
	\For{frame $j = 1$ to $\lceil n^{2/3}\log^{i +1} n  + n^{2/3} \rceil$}
		\State \textbf{for each $c \in Q_{v,i}$ in sequence:} $v$ forwards an unsent message for $c$ to its parent in $c$'s tree.	\label{algStage-i:forward}
	\EndFor
\EndFor
\end{algorithmic}  \label{algStage-i}
\end{algorithm}

\begin{lemma}	\label{lemma:Stage-i}
For all blocker nodes $c \in Q_{v,i}$, node $v$ would have sent $\alpha$ messages to its parent in $c$'s tree by $\alpha + n^{2/3} - h_c (v) $ frames of Stage $i$, where $h_c (v) = hops (v,c)$, provided at least $\alpha$ messages are routed through $v$ in Step~\ref{algStage-i:forward} of Algorithm~\ref{algStage-i}.
\end{lemma}


\begin{proof}
Fix a blocker node $c$.
Let $i'$ be the smallest $i$ for which the above statement does not hold and let $v$ be a node with maximum $h_c (v)$ value for which this
statement is violated in Stage $i'$.
Node $v$ is not a leaf node since
$\alpha$ is 0 or 1 for a leaf and a leaf
would have sent its distance value to its parent in the first frame of Stage-$0$.

So  $v$ must be an internal node.
Since the statement does not hold for $v$ 
for the first time
for $\alpha$, it implies that $v$ has already sent $\alpha -1$ messages (including its own 
distance value $\delta (v,c)$) by $(\alpha - 1) + n^{2/3} -h_c (v)$ frames and now
does not have any message to send to its parent in $c$'s tree in
the next frame.
However since the statement holds for all of $v$'s children, $v$ should have received at least $\alpha -1$ messages from its children by 
$(\alpha -1) + n^{2/3} - (h_c (v) +1)$-th frame, resulting in a contradiction.
\end{proof}

Since $h_c (v) \leq n^{2/3}$, Lemma~\ref{lemma:Stage-i} leads to the following Corollary.

\begin{corollary}	\label{cor:completion}
After the completion of Stage $i$, every node $v$ would have sent all or at least $n^{2/3} \log^{ i +1} n$ different distance values for all blocker nodes $c \in Q_{v,i}$.
\end{corollary}

\begin{lemma}	\label{lemma:Qvj}
The set $Q_{v,i}$ has size at most $\lceil n^{2/3}/\log^{i+1/2} n \rceil$.
\end{lemma}

\begin{proof}
By Corollary~\ref{cor:completion} after the completion of Stage $i-1$, every node $v$ would have sent all or at least $n^{2/3} \log^{i+1} n$ different distance values for all blocker nodes in $Q_{v,i-1}$.
Thus the set $Q_{v,i}$ will consist of only those nodes from $Q$ for 
which $v$ needs to send at least $n^{2/3} \log^{i +1} n$ different distance values.
Since congestion at any node $v$ is at most $n\sqrt{|Q|} = n^{4/3}\log^{1/2} n$ (using Lemma~\ref{lemma:termination}),
the size of $Q_{v,i}$
 is at most $n^{4/3}\log^{1/2} n / n^{2/3}\log^{i+1} n = n^{2/3}/\log^{i+1/2} n$.
 This establishes the lemma.
\end{proof}

\begin{proof}[Proof of Lemma~\ref{lemma:algCase2:round-robin}]
Since $|Q_{v,i}| \leq n^{2/3}/\log^{i+1/2} n$ (by Lemma~\ref{lemma:Qvj}), Stage $i$ runs for 
$n^{2/3}/\log^{i+1/2} n \cdot (n^{2/3}\log^{i +1} n  + n^{2/3}) \leq n^{4/3}\log^{1/2} n + n^{4/3}$ 
rounds.
Lemma~\ref{lemma:algCase2:round-robin} is immediately established from Corollary~\ref{cor:completion} and the fact 
that there are $(1/3)\cdot \log n / \log \log n - 1$ stages.
\end{proof}

\section{Overview of $h$-hop Shortest Path Extension Algorithm}	\label{sec:h-hop-extensions}

We now describe an algorithm for computing $h$-hop \textit{extensions} (Step~\ref{extended} of Algorithm~\ref{algAPSP}) based on the Bellman-Ford algorithm~\cite{Bellman58}.
This algorithm is also used as a step in the randomized APSP algorithm of Huang et al.~\cite{HNS17}.
Here every blocker node $c \in Q$ knows its shortest path distance value from every source node 
$x \in V$ and the 
goal is to extend the shortest path from $x$ to $c$ by additional $h$ hops.

This algorithm works as follows:
Fix a source $x \in V$.
Every blocker node $c \in Q$ initializes the shortest path distance from $x$ to $\delta (x,c)$ (this value is already
known to every $c$).
We then run Bellman-Ford algorithm at every node $v \in V$ for source node
$x$ for $h$ rounds using these initialized 
values.
We repeat this for every $x \in V$.

After this algorithm terminates, every sink node $t \in V$ knows the shortest path distance from every $x \in V$.
Since we run Bellman-Ford for $h$ rounds per source node, 
for each $x\in V$, this whole algorithm takes $O(nh)$ rounds in total. 
This leads to the following lemma.

\begin{lemma}	\label{lemma:h-hop-extension}
The $h$-hop shortest path extensions can be computed in $O(nh)$ rounds for every source $x \in V$ using 
Bellman-Ford algorithm.
\end{lemma}

\section{Conclusion}

We have presented a new deterministic distributed algorithm for computing exact weighted APSP in $\tilde{O}(n^{4/3})$ rounds in both directed and undirected
graphs with arbitrary edge weights.
This algorithm improves on the $\tilde{O}(n^{3/2})$ round APSP algorithm of~\cite{ARKP18}.
At the heart of our algorithm is an efficient distributed algorithm for sending the distance values from source nodes 
to the blocker nodes and an improved
deterministic algorithm for computing the blocker set using pairwise independence and derandomization.
We believe that both these techniques may be of independent interest for obtaining results for other distributed
graph problems.

The main open question left by our work is whether we can get a deterministic algorithm that can match the current $\tilde{O}(n)$ randomized bound for computing 
weighted APSP~\cite{BN19}.

\subsection*{Acknowledgement.}
We thank Valerie King for suggesting using the techniques in Berger et al.~\cite{BRS94} for the blocker set construction.

\bibliographystyle{abbrv}
\bibliography{references}

\appendix
\section{Appendix}	\label{sec:appendix}

\subsection{Broadcast Primitives} \label{sec:broadcast}

In this paper we use the following two broadcast primitives quite extensively.
These primitives are widely known and we restate them here only for completeness.
See~\cite{ARKP18} for more details.

 \begin{lemma}[~\cite{ARKP18}] \label{lem:k-bc}
  A node  $v$ can broadcast  $k$ local values to all other nodes 
    reachable from it
  deterministically in $O(n+k)$ rounds.
  \end{lemma}
  
   \begin{lemma}[~\cite{ARKP18}] \label{lem:all-to-all-bc}
  All $v\in V$ can broadcast a local value  to every other node 
they can reach
 in $O(n)$ rounds deterministically.
  \end{lemma}
    

\subsection{Consistent $h$-hop SSSP ($h$-CSSSP)}	\label{sec:csssp}

The notion of $h$-hop \textit{Consistent SSSP (CSSSP)} was introduced recently in~\cite{AR19}.
The goal of this new notion was to create a consistent collection of paths across all trees in the collection, i.e. a path from $u$ to $v$ is 
same in all trees $T$ in the
CSSSP collection $\mathcal{C}$ (in which such a path exists).

The difference between an $h$-hop SSSP for source $x$ and the tree for source $x$ in the $h$-CSSSP collection $\mathcal{C}$ is that
the former contains path from $x$ to every $t \in V$ for which there exists a path with at most $h$ hops.
However this is not guaranteed in the latter case as the CSSSP only guarantees that if the shortest path from $x$ to $t$ has at most $h$
hops, then this path will be present in the corresponding tree for source $x$ in the CSSSP collection.
This is the major difference between these two notions.
We now re-state the definition of CSSSP from~\cite{AR19} here:

\begin{definition}[{\bf CSSSP}~\cite{AR19}]	
  Let $H$ be a collection of rooted trees of height $h$ in a graph $G=(V,E)$. Then $H$ is an \emph{$h$-CSSSP collection} (or simply an \emph{$h$-CSSSP})
  if for every $u, v \in V$ the path from $u$ to $v$ is 
  the same in each of the trees in $H$ (in which such a path exists), and is the $h$-hop shortest path from $u$ to $v$ in the $h$-hop tree $T_u$ rooted at $u$. 
  Further, each $T_u$ contains every vertex $v$ that has a path with at most $h$ hops from $u$ in $G$ that has distance $\delta(u,v)$.
  \end{definition}

   \cite{AR19} describes a very simple algorithm for constructing $h$-CSSSP collection:
  First compute $2h$-hop SSSPs for every source $x$.
  To compute CSSSP, just retain the initial $h$ hops of each of these $2h$-hop SSSPs.
  We can construct these $h$-hop SSSPs using Bellman-Ford algorithm, leading to the following lemma.
  (See~\cite{AR19} for further details.)
  
  \begin{lemma}[\cite{AR19}]	\label{lemma:CSSSP}
$h$-CSSSPs for source set $S$ can be computed in $O(|S|\cdot h)$ rounds using  the Bellman-Ford algorithm.
\end{lemma}

The CSSSP collection have the following two important properties which we use throughout in this paper.
We call a tree $T$ rooted at a vertex $c$ an out-tree if all the edges incident to $c$ are outgoing edges from $c$ and
 we call $T$ an in-tree if all the edges incident to $c$ are incoming edges.

\begin{lemma}[\cite{AR19}]	\label{lemma:cOutTree}
Let $\mathcal{C}$ be an $h$-CSSSP collection. Let $c$ be a node in $G$ and
let $T$ be the union of the edges in the collection of subtrees rooted at $c$ in the trees in $\mathcal{C}$.
Then $T$ forms an out-tree rooted at $c$.
\end{lemma}

\begin{lemma}[\cite{AR19,ARKP18}]	\label{lemma:cInTree}
 Let $\mathcal{C}$ be an $h$-CSSSP collection. Let $c$ be a node in $G$ and
let $T$ be the union of the edges on the tree-path from the root of each tree in $\mathcal{C}$ to $c$ (for the trees that contain $c$).
Then $T$ forms an in-tree rooted at $c$.
\end{lemma}

\subsection{ $O(n)$ Sample Space Construction}	\label{sec:sample-space}

In this section we describe Luby's approach~\cite{Luby93} to speed up exhaustive search for finding a good 
sample point (or good set $A$, Def.~\ref{def:good-set}) by replacing the sample space $\{0,1\}^n$ of size 
$2^n$ with another sample space of size $O(n)$.
Let $l$ be such that $2n < 2^l \leq 4n$ and let this new sample space be $\{0,1\}^l$.
For each $i \in \{1,n\}$, we consider $i$ as a binary string of length $l$, where the last bit $i_l$ is $1$.
We define $X_i (v) = \oplus_{i=1}^{l} (i_k\cdot z_k)$ for each $i \in \{1,n\}$ and $z \in \{0,1\}^l$, 
where $\oplus$ is addition modulo 2.
For a random string $w$ of length $l$, Luby~\cite{Luby93} showed that the variables $X_1 (w), \ldots, X_n (w)$
are pairwise independent and are identically distributed uniformly in $\{0,1\}$.
This claim allows us to find a good sample point (or good set) in this sample space instead of 
performing an exhaustive search.

\subsection{Correctness Proofs for Algorithm~\ref{algRandBlocker}}	\label{sec:correctness-rand-blocker}

In this Section we provide proofs for our randomized algorithm for computing blocker set for a given $h$-CSSSP collection 
$\mathcal{C}$.
Note that the proof of Lemmas~\ref{lemma:A-size}-\ref{lemma:Pij}, \ref{lemma:selection} and \ref{lemma:BlockerSize} is based on the 
analysis in~\cite{BRS94} and we adapt them here in our setting.
Table~\ref{table-rand-analysis} presents the notation we use in our analysis in this section.

 \begin{table*}
\centering
\caption{List of Notations Used in the Analysis of the Randomized Algorithm} \label{table-rand-analysis} 
\begin{tabular}{| c | l |}
\hline
$Q$ & blocker set (being constructed)\\
\hline
$\mathcal{C}$ & $h$-CSSSP collection	\\
\hline
$S$ & set of sources in $\mathcal{C}$\\
\hline
$h$ & number of hops in a path\\
\hline
$n$ & number of nodes \\
\hline
$V_i$ & set of nodes $v$ with $score(v) \geq (1+\epsilon)^{i-1}$ \\
\hline
$P_i$ & set of paths in $\mathcal{C}$ with at least one node in $V_i$ \\	
\hline
$P_{ij}$ & set of paths in $P_i$ with at least $(1+\epsilon)^{j-1}$ nodes in $V_i$ \\
\hline
$\epsilon, \delta$ & positive constants $\leq 1/12$	\\
\hline
$A$ & set constructed in Step 12 \\
\hline
$X_v$ & $1$ if $v$ is present in $A$, otherwise $0$ \\
\hline 
$Y_1$ & $\sum_{p \in P_i} \sum_{v \in V_i \cap p} X_v$ \\
\hline
$Y_2$ & $\sum_{p \in P_i} \sum_{v,v' \in V_i \cap p} X_{v}\cdot X_{v'}$ \\
\hline
$Y_3$ & $\sum_{p \in P_{ij}} \sum_{v \in V_i \cap p} X_v$ \\
\hline
$Y_4$ & $\sum_{p \in P_{ij}} \sum_{v,v' \in V_i \cap p} X_{v}\cdot X_{v'}$ \\
\hline
$n_{V_i,p}$ & number of nodes from $V_i$ in $p$	\\
\hline
$n_{v, P_{ij}}$ & number of paths in $P_{ij}$ that contain node $v$ \\
\hline
$score(v)$ &  number of root-to-leaf paths in $\mathcal{C}$ that contain $v$ (local var. at $v$)	\\
\hline
$score^{ij} (v)$ &  number of paths in $P_{ij}$ that contain $v$ (local var. at $v$)		\\
\hline
\end{tabular}
\end{table*}

\begin{lemma}	\label{correctness:blocker}
The set $Q$ constructed in Algorithm~\ref{algRandBlocker} is a blocker set for the  CSSSP collection $\mathcal{C}$.
\end{lemma}

\begin{proof}
To show that $Q$ is a blocker set, we need to show that the computed blocker set $Q$ indeed covers all paths in the
 CSSSP collection $\mathcal{C}$.
The while loop in Steps~\ref{startWhile}-\ref{endWhile} runs as long as there is a path in $P_i$ with at least $(1+\epsilon)^{j-1}$ nodes in $V_i$ 
and since this loop terminated for $i=1$ and $j=1$, it implies that there is no path in $\mathcal{C}$ which is not covered by some node in $Q$.
\end{proof}

\begin{lemma}	\label{lemma:Vi}
If the check in Step~\ref{if} fails, then $|V_i| > \frac{(1+\epsilon)^j}{\delta^3}$.
\end{lemma}

\begin{proof}
Since no node in $V_i$ covers a $\frac{\delta^3}{(1+\epsilon)}$ fraction of paths from $P_{ij}$, hence the total $score_{ij}$ values (defined in Step~\ref{computeScoreij})
 for all nodes in $V_i$ has value at most $|V_i|\cdot \frac{\delta^3}{(1+\epsilon)}\cdot |P_{ij}|$.
 And since every path in $P_{ij}$ has at least $(1+\epsilon)^{j-1}$ nodes in $V_i$, 
 $$ |V_i|\cdot \frac{\delta^3}{(1+\epsilon)}\cdot |P_{ij}| > |P_{ij}|\cdot (1+\epsilon)^{j-1}$$
This establishes that $|V_i| > \frac{(1+\epsilon)^j}{\delta^3}$.
\end{proof}

\begin{lemma}	\label{lemma:A-size}
The set $A$ constructed in Step~\ref{pick2} of Algorithm~\ref{algRandBlocker} has size at most 
$(\delta + 2\delta^2)\cdot \frac{|V_i|}{(1+\epsilon)^j}$ and at least $(\delta - 2\delta^2)\cdot \frac{|V_i|}{(1+\epsilon)^j}$  with probability at least $3/4$.
\end{lemma}

\begin{proof}
Consider random variable $X_v$ where $X_v = 1$ if $v$ is present in $A$, otherwise $X_v = 0$.
Thus $\sum_{v \in V_i} X_v$ denotes the size of $A$.
We now calculate its expectation and variance.

\begin{equation}	\label{eq-E[A]}
E[\sum_{v \in V_i} X_v] = |V_i|\cdot \frac{\delta}{(1+\epsilon)^j}
\end{equation}

\begin{equation}	\label{eq-Var(A)}
Var[\sum_{v \in V_i} X_v] = |V_i|\cdot Var[X_v] \leq |V_i|\cdot E[X_v^2] = |V_i|\cdot \frac{\delta}{(1+\epsilon)^j}
\end{equation}

We now use Chebyshev's inequality to get an upper bound on the size of $A$.
Using Chebyshev's inequality the following holds with probability at least $3/4$:

\begin{align*}
||A| -E[|A|]| & \leq 2\sqrt{Var[|A|]}	\\
& \leq 2\sqrt{ |V_i|\cdot \frac{\delta}{(1+\epsilon)^j}} \\
& \leq 2\cdot |V_i| \cdot \frac{\delta^2}{(1+\epsilon)^j} \text{\hspace{.1in} (by Lemma~\ref{lemma:Vi} $\frac{1}{|V_i|} < \frac{\delta^3}{(1+\epsilon)^j}$)} \\
|A| & \leq |V_i|\cdot \frac{\delta}{(1+\epsilon)^j} +  2\cdot |V_i| \cdot \frac{\delta^2}{(1+\epsilon)^j}
\end{align*}

Using the above analysis we can also show that $|A| \geq (\delta - 2\delta^2)\cdot \frac{|V_i|}{(1+\epsilon)^j}$ with probability at least $3/4$.
\end{proof}

\begin{lemma}	\label{lemma:Pi}
The set $A$ constructed in Step~\ref{pick2} of Algorithm~\ref{algRandBlocker} covers at least 
$|A|\cdot (1+\epsilon)^i \cdot (1-3\delta-\epsilon)$ paths in $P_i$ with probability at least $1/2$.
\end{lemma}

\begin{proof}
Consider random variable $X_v$ as defined in the proof of Lemma~\ref{lemma:A-size}.
A path $p$ is covered by $A$ if $v \in A$ for some $v \in V_i \cap p$.
To get a lower bound on the number of paths covered by $A$, we use the term $\sum_{v \in V_i \cap p} X_v - \sum_{v,v' \in V_i \cap p} X_{v}\cdot X_{v'}$ to 
denote if a path $p$ is covered by $A$ or not. 
Note that this term has value at most $1$ which is attained when either $1$ or $2$ nodes from $p$ are picked in $A$ and otherwise the value is non-positive.
Thus the term $\sum_{p \in P_i} [\sum_{v \in V_i \cap p} X_v - \sum_{v,v' \in V_i \cap p} X_{v}\cdot X_{v'}]$ gives a lower bound on the number of paths
covered by $A$ in $P_i$.
Let this term be $Y$.
Now we show that value of $Y$ is $\geq |A|\cdot (1+\epsilon)^i\cdot (1-3\delta-\epsilon)$ with probability at least $1/2$.

We first split $Y$ into $Y_1$ and $Y_2$ where $Y_1 = \sum_{p \in P_i} \sum_{v \in V_i \cap p} X_v$
and $Y_2 = \sum_{p \in P_i} \sum_{v,v' \in V_i \cap p} X_{v}\cdot X_{v'}$.

We first get a lower bound on the term $Y_1$.

\begin{align*}
Y_1 & = \sum_{p \in P_i} \sum_{v \in V_i \cap p} X_v \\
& = \sum_{v\in V_i} \sum_{\{p\in P_i: v\in p\}} X_v \\
& \geq (1+\epsilon)^{i-1}\cdot \sum_{v\in V_i} X_v \text{\hspace{.1in} (since every node in $V_i$ lies in} \\
& \text{\hspace{1.5in} $\geq (1+\epsilon)^{i-1}$ paths in $P_i$)} \\
& = (1+\epsilon)^{i-1} \cdot |A|
\end{align*}

We now need to get an upper bound on the term $Y_2$. 
We first compute an upper bound on $E[Y_2]$ and then use Markov inequality to get an upper bound on $Y_2$.
(Let $n_{V_i,p}$ denotes the number of nodes from $V_i$ in $p$. Clearly $n_{V_i,p} \leq (1+\epsilon)^j$)

\begin{align*}
E[Y_2] & = \sum_{p \in P_i} \sum_{v,v' \in V_i \cap p} E[X_{v}\cdot X_{v'}] \\
& = \sum_{p \in P_i} {n_{V_i,p}\choose 2} \left(\frac{\delta}{(1+\epsilon)^j}\right)^2\\
&\leq  (1+\epsilon)^j \cdot \sum_{p \in P_i} \frac{n_{V_i,p}}{2} \left(\frac{\delta}{(1+\epsilon)^j}\right)^2 \text{\hspace{.1in} (since $n_{V_i,p} \leq (1+\epsilon)^j$)} \\
&\leq (1+\epsilon)^j \cdot \frac{\sum_{v\in V_i} score(v)}{2} \cdot \left(\frac{\delta}{(1+\epsilon)^j}\right)^2 \\
&\leq (1+\epsilon)^j \cdot \frac{|V_i|}{2} \cdot \max_{v \in V_i} score (v) \cdot \left(\frac{\delta}{(1+\epsilon)^j}\right)^2 \\
&\leq \frac{|V_i|}{2} \cdot (1+\epsilon)^{i-j} \cdot \delta^2
\end{align*}

Now using Markov inequality we get the following upper bound on $Y_2$ with probability at least $3/4$:

$$Y_2 \leq 4E[Y_2] \leq 2\delta^2 \cdot  (1+\epsilon)^{i-j} \cdot |V_i|$$

Since $|A| \geq (\delta - 2\delta^2)\cdot \frac{|V_i|}{(1+\epsilon)^j}$ with probability at least $3/4$ by Lemma~\ref{lemma:A-size}, 
$Y_2 \leq 2\delta^2 \cdot (1+\epsilon)^i \cdot \frac{|A|}{(\delta - 2\delta^2)}$ with probability at least $1/2$.

Combining the bounds for $Y_1$ and $Y_2$ we get the following lower bound on $Y$ with probability at least $1/2$:

\begin{align*}
Y & = Y_1 - Y_2 \\
& \geq  (1+\epsilon)^{i-1} \cdot |A| - 2\delta^2 \cdot (1+\epsilon)^i \cdot \frac{|A|}{(\delta - 2\delta^2)} \\
& = (1+\epsilon)^{i} \cdot |A| \cdot (\frac{1}{1+\epsilon} - \frac{2\delta}{1-2\delta})	\\
& = (1+\epsilon)^{i} \cdot |A| \cdot (1 - \frac{\epsilon}{1+\epsilon} - \frac{3\delta}{3/2 -3\delta}) \\
& \geq  (1+\epsilon)^{i} \cdot |A| \cdot (1 -\epsilon - 3\delta)
\end{align*}

This establishes the lemma.
\end{proof}

\begin{lemma}	\label{lemma:Pij}
The set $A$ constructed in Step~\ref{pick2} of Algorithm~\ref{algRandBlocker} covers at least a $\delta/2$ fraction of paths in $P_{ij}$ with probability at least 
$5/8$.
\end{lemma}

\begin{proof}
Similar to the proof of Lemma~\ref{lemma:Pi} we can lower bound the number of paths covered by set $A$ in $P_{ij}$ by the term \\
$\sum_{p \in P_{ij}} [\sum_{v \in V_i \cap p} X_v - \sum_{v,v' \in V_i \cap p} X_{v}\cdot X_{v'}]$.
Let this term be $Y'$, with first term $Y_3$ and the second term $Y_4$.
We need to show that $Y' \geq \frac{\delta}{2}\cdot |P_{ij}|$ with probability at least $5/8$.

We first give a lower bound on $Y_3$.
To get the lower bound,  we first compute a lower bound on $E[Y_3]$ and an upper bound on $Var[Y_3]$ and then use Chebyshev's inequality.
(Let $n_{v,P_{ij}}$ represent the number of paths in $P_{ij}$ that contain node $v$.
Since no node covers at least $\frac{\delta^3}{(1+\epsilon)}$ fraction of paths in $P_{ij}$, $n_{v,P_{ij}} < \frac{\delta^3}{(1+\epsilon)}$ )

\begin{align*}
E[Y_3] & = E[\sum_{p \in P_{ij}} \sum_{v\in V_i \cap p} X_v] \\
& \geq E[\sum_{p \in P_{ij}} (1+\epsilon)^{j-1} \cdot X_v] \text{\hspace{.1in} (since every path in $P_{ij}$ has }\\
& \text{\hspace{2in} $\geq (1+\epsilon)^{j-1}$ nodes from $V_i$)} \\ 
& = (1+\epsilon)^{j-1} \cdot |P_{ij}| \cdot \frac{\delta}{(1+\epsilon)^j} \\
& =  |P_{ij}| \cdot \frac{\delta}{(1+\epsilon)}
\end{align*}

\begin{align*}
Var[Y_3] & = Var[\sum_{p\in P_{ij}} \sum_{v\in V_i\cap p} X_v] \\
& = Var[\sum_{v\in V_i} \sum_{\{p\in P_{ij}: v\in p\}} X_v] \\
& = Var[\sum_{v\in V_i} n_{v,P_{ij}} X_v] \\
& = \sum_{v\in V_i} n^2_{v,P_{ij}} \cdot Var[X_v] \\
& \leq \frac{\delta}{(1+\epsilon)^j} \cdot \frac{\delta^3}{(1+\epsilon)}\cdot |P_{ij}| \cdot \sum_{v\in V_i} n_{v,P_{ij}} \text{\hspace{.1in} (since }\\
& \text{\hspace{2.5in} $n_{v,P_{ij}} < \frac{\delta^3}{(1+\epsilon)}\cdot |P_{ij}| $)} \\
& \leq \frac{\delta^4}{(1+\epsilon)^{j+1}} \cdot |P_{ij}| \cdot |P_{ij}| \cdot (1+\epsilon)^j \text{\hspace{.1in} (since every path in}\\
& \text{\hspace{2.5in} $P_{ij}$ has $\leq (1+\epsilon)^{j}$ nodes from $V_i$)} \\ 
& \leq \delta^4\cdot |P_{ij}|^2
\end{align*}

We now use Chebyshev's inequality to get a lower bound on the value of $Y_3$.
Using Chebyshev's inequality the following holds with probability at least $7/8$:

\begin{align*}
|Y_3 -E[Y_3]| & \leq 2\sqrt{2}\sqrt{Var[Y_3]}	\\
Y_3 & \geq E[Y_3] - 2\sqrt{2} \delta^2 \cdot |P_{ij}| \\
& \geq |P_{ij}| \cdot \frac{\delta}{(1+\epsilon)} - 2\sqrt{2} \delta^2 \cdot |P_{ij}|  
\end{align*}

We now need to get an upper bound on the term $Y_4$. 
We first compute an upper bound on $E[Y_4]$ and then use Markov inequality to get an upper bound on $Y_4$.

\begin{align*}
E[Y_4] & = \sum_{p \in P_{ij}} \sum_{v,v' \in V_i \cap p} E[X_{v}\cdot X_{v'}] \\
& \leq |P_{ij}|\cdot \frac{(1+\epsilon)^{2j}}{2} \cdot \left(\frac{\delta}{(1+\epsilon)^j}\right)^2 \text{\hspace{.1in} (since there are $\leq (1+\epsilon)^j$} \\
& \text{\hspace{2.5in}  nodes from $V_i$ in any path in $P_{ij}$)} \\
& =  |P_{ij}|\cdot \frac{\delta^2}{2}
\end{align*}

Now using Markov inequality we get the following upper bound on $Y_4$ with probability at least $3/4$:

$$Y_4 \leq 4E[Y_4] \leq 2\delta^2 \cdot |P_{ij}|$$

Combining the bounds for $Y_3$ and $Y_4$ we get the following lower bound on $Y'$ with probability at least $5/8$:

\begin{align*}
Y' & = Y_3 - Y_4 \\
&\geq |P_{ij}| \cdot \frac{\delta}{(1+\epsilon)} - 2\sqrt{2} \delta^2 \cdot |P_{ij}| - 2\delta^2 \cdot |P_{ij}| \\
&\geq |P_{ij}| \cdot \delta \cdot (1-\epsilon -5\delta) \\
&\geq |P_{ij}| \cdot \frac{\delta}{2} \text{\hspace{.1in} (since $\epsilon, \delta \leq 1/12$)} 
\end{align*}

This establishes the lemma.
\end{proof}

\begin{lemma*}{\bf \ref{lemma:A}}
The set $A$ constructed in Step~\ref{pick2} of Algorithm~\ref{algRandBlocker} is a good set with probability at least $\frac{1}{8}$, i.e. $A$ covers at least $|A|\cdot (1+\epsilon)^i \cdot (1-3\delta-\epsilon)$ paths in $P_i$, including at least a $\delta/2$ fraction of paths in $P_{ij}$.
\end{lemma*}

\begin{proof}
This is immediate from Lemma~\ref{lemma:Pi} and \ref{lemma:Pij}. 
\end{proof}

\begin{lemma*}{\bf \ref{lemma:selection}}
The number of selection steps is $O\left(\frac{\log^3 n}{\delta^3\cdot \epsilon^2} \right)$, i.e. the while loop in Steps~\ref{startWhile}-\ref{endWhile} runs for at most $O\left(\frac{\log^3 n}{\delta^3\cdot \epsilon^2} \right)$ iterations in total.
\end{lemma*}

\begin{proof}
The while loop runs until $P_{ij}$ is non-empty, i.e. there exists a path in $P_i$ with at least $(1+\epsilon)^{j-1}$ nodes in $V_i$.
In each iteration, the algorithm either covers at least $\frac{\delta^3}{(1+\epsilon)}$ fraction of paths in $P_{ij}$ (if node $c$ is added to blocker set $Q$
 in Step~\ref{algRandBlocker:pick1}) or at least $\frac{\delta}{2}$ fraction of paths from $P_{ij}$ (if set $A$ is added to $Q$ in Step~\ref{algRandBlocker:checkA}).
 Since there are at most $n^2$ paths and each iteration of the while loop covers at least $\frac{\delta^3}{(1+\epsilon)}$ fraction of $P_{ij}$, there are at most
 $O\left(\frac{\log n^2}{\log \left(\frac{1}{1 - \frac{\delta^3}{(1+\epsilon)}}\right)}\right) = O\left(\frac{(1+\epsilon)\log n}{\delta^3}\right) = O\left(\frac{\log n}{\delta^3}\right)$ iterations.
 Since both the inner and outer for loop runs for $O(\log_{1+\epsilon} n) = O(\frac{\log n}{\epsilon})$ iterations,
this establishes the lemma.
\end{proof}

\begin{lemma}	\label{innerForloop}
Each iteration of the inner for loop (Steps~\ref{startFor2}-\ref{endFor2}) in Algorithm~\ref{algRandBlocker} takes $\tilde{O}\left(\frac{|S|\cdot h}{\delta^3}\right)$ rounds in expectation.
\end{lemma}

\begin{proof}
We first show that each iteration of the while loop in Steps~\ref{startWhile}-\ref{endWhile} takes 
$O(|S|\cdot h)$ rounds in expectation.
Step~\ref{Pij} takes $O(|S|\cdot h)$ rounds by Lemmas~\ref{lemma:compute-Pij} and \ref{lemma:compute-Pij}
and so does Step~\ref{computeScoreij}~\cite{ARKP18} and by Lemma~\ref{lem:all-to-all-bc}.
The check in Step~\ref{if} involves no communication and so does Step~\ref{algRandBlocker:pick1}, since every node knows the $score_{ij}$ values for every other 
node and also the value of $|P_{ij}|$, i.e. the number of paths that belong to $P_{ij}$.
Steps~\ref{pick2} and \ref{algRandBlocker:checkA} are also local steps and does not involve any communication.
Step~\ref{algRandBlocker:broadcastA} involves broadcasting at most $n$ messages and hence takes $O(n)$ rounds using Lemma~\ref{lem:all-to-all-bc}.
Since by Lemma~\ref{lemma:A} the set $A$ constructed in Step~\ref{pick2} is good with probability at least $1/8$, Steps~\ref{pick2}-\ref{algRandBlocker:checkA}
are executed $O(1)$ times in expectation.
Step~\ref{reconstruct} takes $O(|S|\cdot h)$ rounds~\cite{ARKP18} and using Lemma~\ref{lemma:compute-Pi}.
Since the while loop runs for at most $O\left(\frac{\log n}{\delta^3}\right)$ iterations (by Lemma~\ref{lemma:selection}), this establishes the lemma.
\end{proof}

\subsection{Helper Algorithms for Deterministic Blocker Set Algorithm: Distributed Computation of Terms $\nu_{P_i}$ and $\nu_{P_{ij}}$}	\label{sec:helper2'}	\label{sec:B}

In this Section we describe a
simple pipelined
 algorithm to compute $\nu_{P_i}$ and $\nu_{P_{ij}}$ terms at leader node $l$.
Both algorithms are similar to an algorithm in~\cite{ARKP18} (for computing `initial scores').
Recall that $\sigma_{P_i,v}^{(\mu)}$ refers to the number of paths in $P_i^v$ covered by the sample point $X^{(\mu)}$
and $\sigma_{P_{ij},v}^{(\mu)}$ refers to the total number of paths in $P_{ij}^v$ covered by the sample point $X^{(\mu)}$.
Let 
$\nu_{P_i, v}^{(\mu)}$ refers to the sum total of the $\sigma_{P_i,w}^{(\mu)}$ values of all descendant nodes $w$ of $v$ and similarly let 
$\nu_{P_{ij},v}^{(\mu)}$ refers to the 
sum total of the $\sigma_{P_{ij},w}^{(\mu)}$ values of all descendant nodes $w$ of $v$.
Also recall from Section~\ref{sec:deterministic-blocker} 
that $\nu^{(\mu)}_{P_i}$ and  $\nu^{(\mu)}_{P_{ij}}$ refers to the total number of paths covered by $X^{(\mu)}$
in sets $P_i$ and $P_{ij}$ respectively.
Table~\ref{table-det} presents the notations that we use in this Section.

 \begin{table*}
\centering
\caption{List of Notations Used in the Analysis of the Deterministic Algorithm} \label{table-det} 
\begin{tabular}{| c | l |}
\hline
$A$ & set constructed in Step 12 of Randomized Blocker Set Algorithm (Alg.~\ref{algRandBlocker})\\
\hline
$X_v$ & $1$ if $v$ is present in $A$, otherwise $0$ \\
\hline
$X$ & vector composed of $X_v$'s \\
\hline
$X^{(\mu)}$ & $\mu$-th vector in the enumeration of $X$ in the sample space \\
\hline 
$S$ & set of sources in $\mathcal{C}$\\
\hline
$h$ & number of hops in a path\\
\hline
$n$ & number of nodes \\
\hline
$V_i$ & set of nodes $v$ with $score(v) \geq (1+\epsilon)^{i-1}$ \\
\hline
$P_i$ & set of paths in $\mathcal{C}$ with at least one node in $V_i$ \\	
\hline
$P_{ij}$ & set of paths in $P_i$ with at least $(1+\epsilon)^{j-1}$ nodes in $V_i$ \\
\hline
$P_i^{u}$ & set of paths in $P_i$ with leaf node $u$ \\
\hline
$P_{ij}^{u}$ & set of paths in $P_{ij}$ with leaf node $u$ \\
\hline
$\sigma_{P_i,u}$ & $\sum_{p \in P_i^u} \vee_{v \in V_i \cap p} X_{v}$ \\
\hline
$\sigma_{P_{ij},u}$ & $\sum_{p \in P_{ij}^u} \vee_{v \in V_i \cap p} X_v$ \\
\hline
$\nu_{P_i,u}$ & sum total of $\sigma_{P_i,w}$ values for all descendant nodes $w$ of $v$ \\
\hline
$\nu_{P_{ij},u}$ & sum total of $\sigma_{P_{ij},w}$ values for all descendant nodes $w$ of $v$ \\
\hline
$\nu^{(\mu)}_{P_i}$ & value of $\nu_{P_i}$ with $X^{(\mu)}$ as the input \\
\hline
$\nu^{(\mu)}_{P_{ij}}$ & value of $\nu_{P_{ij}}$ with $X^{(\mu)}$ as the input \\
\hline
$\nu^{(\mu)}_{P_i,u}$ & value of $\nu_{P_i, u}$ with $X^{(\mu)}$ as the input \\
\hline
$\nu^{(\mu)}_{P_{ij},u}$ & value of $\nu_{P_{ij}, u}$ with $X^{(\mu)}$ as the input \\
\hline
$\mathcal{C}$ & $h$-hop CSSSP collection	\\
\hline
$T_x$ & $h$-hop shortest path tree rooted at $x$ in collection $\mathcal{C}$ \\
\hline
$Q$ & blocker set (being constructed)\\
\hline
$l$ & leader node \\
\hline
\end{tabular}
\end{table*}

  \subsubsection{Computing $\nu_{P_{i}}$}	\label{sec:Y4}

Consider 
computing the $\nu_{P_i}^{(\mu)}$ terms for each sample point $\mu$, at leader node $l$ (Algorithm~\ref{algY4})
($\nu_{P_{ij}}$ can be computed similarly).
First every node $v$ initializes its $\nu_{P_i, v}^{(\mu)}$ value, for each sample point $\mu$, in Step~\ref{algY4:init1}. 
Recall that we assume that all $X$ values are enumerated in order and every node knows this enumeration.
In round $n -1 -h +\mu$, the node $u$ at height $h$ sends its corresponding $\nu_{P_i, u}$ value for $X^{(\mu)}$ (Step~\ref{algY4:send}) 
along with the total value of $\nu_{P_i}$ it 
received from its children for $X^{(\mu)}$ (Steps~\ref{algY4:new-receiveStart}-\ref{algY4:receiveEnd}).
Leader node $l$ then computes the total sum $\nu_{P_i}^{(\mu)}$ for each sample point $\mu$, by summing up the received 
$\nu_{P_i,w}^{(\mu)}$ values from all its children $w$ in Step~\ref{algY4:local}.
 In Lemma~\ref{lemma:Y4} we show that leader $l$ correctly computes $\nu_{P_i}$ values for all $X^{(\mu)}$'s in $O(n)$ rounds.

\begin{algorithm}	
\caption{Compute sum of $\nu_{P_i}$ values at leader node $l$}
\begin{algorithmic}[1]	
\Statex Input: $h$: number of hops; $S$: set of sources; $\mathcal{C}$: $h$-CSSSP collection; $X^{(\mu)}$: $\mu$-th vector in sample space; $T$: BFS in-tree rooted at leader $l$
\State \textbf{Local Step at $v \in V$:} Let $\mathcal{P}$ be the set of paths in $P_{i}$ with $v$ as the leaf node. \textbf{For each $1 \leq \mu \leq n$}, set $\nu^{(\mu)}_{P_i,v} = \sum_{p\in \mathcal{P}} \vee_{z \in p} X^{(\mu)}_z$ \vspace{.05in} \label{algY4:init1}
\State {\bf In round $r > 0$ (for all nodes $v \in V - \{t\}$):} \vspace{.05in}
\State {\bf send:} {\bf if} $r = n -h(v) + \mu -1$ {\bf then} send $\langle \nu^{(\mu)}_{P_i,v} \rangle$ to $parent (v)$ in $T$	\label{algY4:send}
\State {\bf receive [lines~\ref{algY4:new-receiveStart}-\ref{algY4:receiveEnd}]:}  
\If{$r = n - h(v) + \mu - 2$}	\label{algY4:new-receiveStart}
	\State let $\mathcal{I}$ be the set of incoming messages to $v$ 
	\For{{\bf each} $M \in \mathcal{I}$} \label{algY4:receiveStart}
		\State let the sender be $w$ and let $M = \langle \nu_{P_i,w}^{(\mu)} \rangle$ and 
	 	\State {\bf if} $w$ is a child of $v$ in $T$ {\bf then} $\nu^{(\mu)}_{P_i,v} \leftarrow \nu^{(\mu)}_{P_i,v} + \nu_{P_i,w}^{(\mu)}$ 
	\EndFor \label{algY4:receiveEnd}
\EndIf
\State \textbf{Local Step at leader $l$:} Compute the total sum $\nu_{P_i}^{(\mu)}$ for each sample point $\mu$, by summing up the received $\nu_{P_i,w}^{(\mu)}$ values from all its children $w$.	\label{algY4:local}
\end{algorithmic} \label{algY4}
\end{algorithm}

\begin{lemma}	\label{lemma:Y4}
Algorithm~\ref{algY4} correctly computes the $\nu_{P_i}^{(\mu)}$ values at leader node $l$ for all $\mu$ in $O(n)$ rounds.
\end{lemma}

\begin{proof}
In Step~\ref{algY4:init1}, every node $v$ correctly initialize their contribution to the overall $\nu_{P_i,v}$ term for each $\mu$ locally.
Since the height of tree $T$ is at most $n-1$, it is readily seen that a node $v$ that is at depth $h(v)$ in $T$ will receive the $count^{(\mu)}_{P_i}$ 
values from its children in round $n -h(v) + \mu-2$ (Steps~\ref{algY4:new-receiveStart}-\ref{algY4:receiveEnd})
and thus will have the correct $\nu^{(\mu)}_{P_i}$  value to send in round $n -h(v) + \mu -1$ 
in Step~\ref{algY4:send}.
Since $\mu = O(n)$,
  Steps~\ref{algY4:send}-\ref{algY4:receiveEnd} runs 
  in $O(n)$ rounds.
Step~\ref{algY4:local} is a local step and thus does not involve any communication.
This establishes the lemma.
\end{proof}


 \subsubsection{Computing $\nu_{P_{ij}}$}	\label{sec:Y2}
  
   Here we  describe our algorithm for computing $\nu_{P_{ij}}^{(\mu)}$ terms for each sample point $\mu$, at leader node $l$ (Algorithm~\ref{algY2}).
Every node $v$ first initializes its $\nu_{P_{ij}}^{(\mu)}$ value in Step~\ref{algY2:init1}. 
Recall that we assume that all $X$ values are enumerated in order and every node knows this enumeration.
In round $n -1 -h +\mu$, the node $u$ at height $h$ sends its corresponding $\nu_{P_{ij}, u}$ value for $X^{(\mu)}$ (Step~\ref{algY2:send}) 
along with the total value of $\nu_{P_{ij}}$ it 
received from its children for $X^{(\mu)}$ (Steps~\ref{algY2:new-receiveStart}-\ref{algY2:receiveEnd}).
Leader node $l$ then computes the total sum $\nu_{P_i}^{(\mu)}$ for each sample point $\mu$, by summing up the received 
$\nu_{P_i,w}^{(\mu)}$ values from all its children $w$ in Step~\ref{algY2:local}.
In Lemma~\ref{lemma:Y2} we show that leader $l$ correctly computes $\nu_{P_{ij}}$ values for all $X^{(\mu)}$'s in $O(n)$ rounds.

\begin{algorithm}	
\caption{Compute sum of $\nu_{P_{ij}}$ values at leader node $l$}
\begin{algorithmic}[1]	
\Statex Input: $h$: number of hops; $S$: set of sources; $\mathcal{C}$: $h$-CSSSP collection; $X^{(\mu)}$: $\mu$-th vector in sample space; $T$: BFS in-tree rooted at leader $l$
\State \textbf{Local Step at $v \in V$:} Let $\mathcal{P}$ be the set of paths in $P_{ij}$ with $v$ as the leaf node. \textbf{For each $1 \leq \mu \leq n$}, set $\nu^{(\mu)}_{P_{ij},v} = \sum_{p\in \mathcal{P}} \vee_{z \in p} X^{(\mu)}_z$ \label{algY2:init1}
\State {\bf In round $r > 0$:} \vspace{.05in}
\State {\bf send:} {\bf if} $r = n -h(v) + \mu -1$ {\bf then} send $\langle \nu^{(\mu)}_{P_{ij},v} \rangle$ to $parent (v)$ in $T$	\label{algY2:send}
\State {\bf receive [lines~\ref{algY4:new-receiveStart}-\ref{algY4:receiveEnd}]:}  
\If{$r = n - h(v) + \mu - 2$}	\label{algY2:new-receiveStart}
	\State let $\mathcal{I}$ be the set of incoming messages to $v$ 
	\For{{\bf each} $M \in \mathcal{I}$} \label{algY2:receiveStart}
		\State let the sender be $w$ and let $M = \langle \nu_{P_{ij},w}^{(\mu)} \rangle$ and 
	 	\State {\bf if} $w$ is a child of $v$ in $T$ {\bf then} $\nu^{(\mu)}_{P_{ij},v} \leftarrow \nu^{(\mu)}_{P_{ij},v} + \nu_{P_{ij},w}^{(\mu)}$ 
	\EndFor \label{algY2:receiveEnd}
\EndIf
\State \textbf{Local Step at leader $l$:} Compute the total sum $\nu_{P_{ij}}^{(\mu)}$ for each sample point $\mu$, by summing up the received $\nu_{P_{ij},w}^{(\mu)}$ values from all its children $w$.	\label{algY2:local}
\end{algorithmic} \label{algY2}
\end{algorithm}

\begin{lemma}	\label{lemma:Y2}
Algorithm~\ref{algY2} correctly computes the $\nu_{P_{ij}}^{(\mu)}$ values at leader node $l$ for all $\mu$ in $O(n)$ rounds.
\end{lemma}

\begin{proof}
In Step~\ref{algY2:init1}, every node $v$ correctly initialize their contribution to the overall $\nu_{P_{ij}}$ term for each $\mu$ locally.
Since the height of tree $T$ is at most $n-1$, it is readily seen that a node $v$ that is at depth $h(v)$ in $T$ will receive the $\nu_{P_{ij}}^{(\mu)}$ 
values from its children in round $n -h(v) +\mu -2$ (Steps~\ref{algY2:new-receiveStart}-\ref{algY2:receiveEnd})
and thus will have the correct $\nu_{P_{ij}}^{(\mu)}$  value to send in round $n -h(v) + \mu -1$ 
in Step~\ref{algY2:send}.
Since $\mu = O(n)$, Steps~\ref{algY2:send}-\ref{algY2:receiveEnd} runs for at most $2n$ rounds.
Step~\ref{algY2:local} is a local step and thus does not involve any communication.
This establishes the lemma.
\end{proof}

\begin{proof}[Proof of Lemma~\ref{lemma:deterministic}]
Every node $v$ correctly computes all the ancestor nodes in each tree $T_x$ in Step~\ref{algDet:Anc} using Algorithm 4 in~\cite{ARKP18} that takes $O(|S|\cdot h)$ rounds~\cite{ARKP18}.
Step~\ref{algDet:init} computes the incoming BFS tree rooted at leader node $l$ in $O(n)$ rounds.
Step~\ref{algDet:compute} takes $O(n)$ rounds by Lemmas~\ref{lemma:Y4} and \ref{lemma:Y2} .
Step~\ref{algDet:B} is a local step and involves no communication.
Step~\ref{algDet:broadcast} involves an all-to-all broadcast of at most $n$ messages and thus takes $O(n)$ rounds using Lemma~\ref{lem:all-to-all-bc}.
\end{proof}

\subsection{Helper Algorithms for Algorithm~\ref{algCase2}}

\subsubsection{Computing Bottleneck Nodes}	\label{sec:bottleneck}

Here we describe our deterministic algorithm for 
computing Step 1 of Algorithm~\ref{algCase2}, which
 identifies a set $B$ of bottleneck nodes such that removing this set of nodes reduces
the congestion in the network from $O(n\cdot |Q|)$ to $O(n\cdot \sqrt{|Q|})$.
However when randomization is allowed, there is a $O(n\cdot \sqrt{|Q|})$ randomized algorithm of Huang et al.~\cite{HNS17} that computes
this set w.h.p. in $n$.
Our deterministic algorithm is however very different from the randomized algorithm given in~\cite{HNS17} and it uses ideas from our
blocker set algorithm in~\cite{ARKP18}.

We now give an overview of the randomized algorithm of~\cite{HNS17} that computes this set of bottleneck nodes.
For a source $x$ and its incoming shortest path tree $T_x$, every node in $T_x$ calculates the number of outgoing messages for source 
$x$.
This is done by waiting for messages from all children nodes, followed by sending a message to its parent in $T_x$.
This takes $O(n)$ rounds and can be run across multiple  nodes in $Q$ as congestion is at most $O(|Q|)$.
Thus using the randomized algorithm of Ghaffari~\cite{Ghaffari15}, this algorithm can be run across all nodes in $Q$ concurrently in 
$\tilde{O}(n + |Q|) = \tilde{O}(n)$ rounds.
After computing these values, a node $b$ with maximum count is selected to the set $B$ and is then removed from the network.
The algorithm repeats this for $O(\sqrt{|Q|})$ times, thus eliminating all nodes that needed to send at least $n\sqrt{|Q|}$ 
messages (since removal of every such node eliminates $O(n\sqrt{|Q|})$ nodes across all trees and there are at most $n\cdot |Q|$
nodes).

Our deterministic algorithm for computing bottleneck nodes (Algorithm~\ref{algBottleneck}) works as follows:
In Step~\ref{algBottleneck:Step1}, the algorithm computes the $count_{v,c}$ values (number of messages $v$ needs to send to its parent
in $c$'s tree) using Algorithm~\ref{algCount} described in Section~\ref{sec:count}.
Every node $v$ calculates the total number of messages it needs to send by summing up the values computed in Step~\ref{algBottleneck:Step1} (Step~\ref{algBottleneck:local-compute}) and then broadcast this value in Step~\ref{algBottleneck:broadcast}.
The node with maximum value is added to the bottleneck node set $B$ (Step~\ref{algBottleneck:add-b}) and the values of its ancestors and descendants are updated 
using the algorithms in~\cite{AR19}.
In Lemma~\ref{lemma:algBottleneck} we establish that the whole algorithm runs in $O(n\sqrt{|Q|} + h\cdot |Q|)$ rounds deterministically.

\begin{algorithm}
\caption{{\sc Compute-Bottleneck:} Compute Bottleneck Nodes Set $B$}
\begin{algorithmic}[1]
\Statex Input: $Q$: blocker set; $\mathcal{C}^Q$: CSSSP collection for blocker set $Q$
\Statex Output: $B$: set of bottleneck nodes
\State \textbf{For each $c \in Q$ in sequence:} Compute $count_{v,c}$ values at every node $v \in V$ using Algorithm~\ref{algCount} (Section~\ref{sec:count}). 	\label{algBottleneck:Step1}
\State \textbf{Local Step at $v \in V$:} Compute $total\_count_v \leftarrow \sum_{c\in Q} count_{v,c}$	\label{algBottleneck:local-compute}
\While{{\bf there is a node $v$ with $\bm{total\_count_v > n\sqrt{|Q|}}$}}	\label{algBottleneck:start-While}
	\State \textbf{For each $v \in V$:} Broadcast $ID(v)$ and $total\_count_v$ value	.	\label{algBottleneck:broadcast}
	\State Add node $b$ to $B$ such that $b$ has maximum $total\_count_v$ value (break ties using IDs).	\label{algBottleneck:add-b}
	\State Update $total\_count_v$ values for the descendants and ancestors of $b$ across all trees in the collection $\mathcal{C}^Q$ using results in~\cite{ARKP18,AR19}.	\label{algBottleneck:update}
\EndWhile		\label{algBottleneck:end-while}
\end{algorithmic}  \label{algBottleneck}
\end{algorithm}

\begin{lemma}	\label{lemma:termination}
After {\sc Compute-Bottleneck} (Algorithm~\ref{algBottleneck}) terminates, $total\_count_v \leq n\sqrt{|Q|}$ for all nodes $v$.
\end{lemma}

\begin{proof}
This is immediate since the while loop in Steps~\ref{algBottleneck:start-While}-\ref{algBottleneck:end-while} terminates only when there is
no node $v$ with $total\_count_v > n\sqrt{|Q|}$.
\end{proof}

\begin{lemma}	\label{lemma:Bottleneck}
The set of bottleneck nodes, $B$, constructed by {\sc Compute-Bottleneck} (Algorithm~\ref{algBottleneck}) has size at most $\sqrt{|Q|}$.
\end{lemma}

\begin{proof}
Since every node $b$ added to set $B$ has $total\_count_b > n\sqrt{|Q|}$, removing such $b$ is going to remove at least $n\sqrt{|Q|}$
nodes across all trees in $Q$ in Step~\ref{algBottleneck:update}.
And since there are at most $n\cdot |Q|$ nodes across all trees, set $B$ has size at most $\sqrt{|Q|}$.
\end{proof}

\begin{lemma}	\label{lemma:algBottleneck}
{\sc Compute-Bottleneck} (Algorithm~\ref{algBottleneck}) runs for $O(n\sqrt{|Q|} + h\cdot |Q|)$ rounds.
\end{lemma}

\begin{proof}
Step~\ref{algBottleneck:Step1} takes $O(h\cdot |Q|)$ rounds using Lemma~\ref{lemma:algCount}.
Step~\ref{algBottleneck:local-compute} is a local computation step and involves no communication.
Step~\ref{algBottleneck:broadcast} involves a broadcast of at most $n$ messages and hence takes $O(n)$ rounds using Lemma~\ref{lem:all-to-all-bc}.
Step~\ref{algBottleneck:add-b} again do not involve any communication.
Step~\ref{algBottleneck:update} takes $O(n)$ rounds~\cite{ARKP18,AR19}.
Since $B$ has size at most $\sqrt{|Q|}$ (by Lemma~\ref{lemma:Bottleneck}), the while loop runs for at most $\sqrt{|Q|}$ iterations, thus 
establishing the lemma.
\end{proof}

\subsubsection{Computing $count_{v,c}$ Values}	\label{sec:count}

\begin{algorithm}
\caption{{\sc Compute-Count:} Algorithm for computing $count_{v,c}$ values for source $c$ at node $v$}
\begin{algorithmic}[1]
\Statex Input: $h$: number of hops, $T_c$: tree for source $c$
\State \textbf{(Round $0$):} \textbf{if $v \in T_c$ then} set $count_{v,c} \leftarrow 1$ \textbf{else} $count_{v,c} \leftarrow 0$		\label{algCount:Step1}
\State \textbf{Round $h+1\geq r > 0$:}
\State \textbf{Send:} \textbf{if} $r = h -h_c (v) + 1$ \textbf{then} send $\langle count_{v,c} \rangle$ to $v$'s parent \vspace{.02in}	\label{algCount:send}
\State {\bf receive [lines~\ref{algCount:receive-Start}-\ref{algCount:receive-End}]:}  
\If{$r = h -h_c (v)$}	\label{algCount:receive-Start}
	\State let $\mathcal{I}$ be the set of  incoming message to $v$ 
	\For{$M \in \mathcal{I}$}
		\State let the sender be $w$ and let $M = \langle count_{w,c} \rangle$ and 
		\State {\bf if} $w$ is a child of $v$ in $T_c$ {\bf then} $count_{v,c} \leftarrow count_{v,c} + count_{w,c}$ 	\label{algCount:update}
	\EndFor
\EndIf	\label{algCount:receive-End}
\end{algorithmic}  \label{algCount}
\end{algorithm} 

Here we describe our algorithm for computing Step 1 of Algorithm~\ref{algBottleneck}, 
which computes $count_{v,c}$ values in a given $h$-CSSSP collection 
$\mathcal{C}$ for source set $S$.
Our algorithm (Algorithm~\ref{algCount}) is quite simple and works as follows:
Fix a source $c \in S$ and let $T_c$ be the tree corresponding to source $c$ in $\mathcal{C}$.
The goal is to compute the number of messages each node $v \in T_c$ needs to send to its parent.
In Step~\ref{algCount:Step1} every node $v \in T_c$ initializes its $count_{v,c}$ value to $1$.
Every node $v$ that is $h_c (v)$ hops away from $c$ receives the $count$ values from all its children by round 
$h - h_c (v)$ (Steps~\ref{algCount:receive-Start}-\ref{algCount:receive-End}) and it then send it to its parent in 
round $h - h_c (v) + 1$ (Step~\ref{algCount:send}) after updating it (Step~\ref{algCount:update}).

\begin{lemma}	\label{lemma:algCount}
{\sc Compute-Count} (Algorithm~\ref{algCount}) correctly computes $count_{v,c}$ for every $v \in T_c$ in $h+1$ rounds per source node $c$.
 \end{lemma}
 
 \begin{proof}
 Every leaf node $v$ can initialize their $count_{v,c}$ values to $1$ in Step~\ref{algCount:Step1}.
 For every other internal node $v$, $v$ correctly computes $count_{v,c}$ value after receiving the $count$ values from all its children 
 by round $h -h_c(v)$ (Steps~\ref{algCount:receive-Start}-\ref{algCount:receive-End})
 and then send the correct $count_{v,c}$ value to its parent in round $h - h_c(v) +1$ in Step~\ref{algCount:send}.
 
 Since $h_c (v) \geq 0$, this algorithm requires at most $h+1$ rounds. 
  \end{proof}

\end{document}